\documentclass[11pt]{article}% APA Reference Style 
\usepackage[a4paper,margin=1in]{geometry}
\usepackage{graphicx}%
\usepackage{multirow}%
\usepackage{amsmath,amssymb,amsfonts}%
\usepackage{amsthm}%
\usepackage{mathrsfs}%
\usepackage[title]{appendix}%
\usepackage{xcolor}%
\usepackage{textcomp}%
\usepackage{manyfoot}%
\usepackage{booktabs}%
\usepackage{algorithm}%
\usepackage{algorithmicx}%
\usepackage{algpseudocode}%
\usepackage{listings}%
\usepackage{mathtools}
\usepackage{multirow}
\usepackage {booktabs}
\usepackage{lscape}
\usepackage{url}
\usepackage[hang,small,bf]{caption}
\usepackage[subrefformat=parens]{subcaption}
\usepackage{amsthm}
\usepackage{rotating}
\usepackage{enumerate}%[A1]
\newtheorem{theorem}{Theorem}%  meant for continuous numbers
\newtheorem{remark}{Remark}%

\newtheorem{definition}{Definition}%

\usepackage[authoryear,round]{natbib}
\newcommand{\E}{\mathrm{E}}

\newcommand{\wtk}{\mathrm{WTK}}
\newcommand{\wsk}{\mathrm{WSK}}
\newcommand{\MISE}{\mathrm{MISE}}
\newcommand{\ISB}{\mathrm{ISB}}
\newcommand{\IV}{\mathrm{IV}}
\newcommand{\cv}{\mathrm{CV}}
\newcommand{\er}{\mathrm{ER}}
\newcommand{\sgn}{\mathrm{sgn}}
\newcommand{\argmin}{\mathop{\rm arg~min}\limits}

\newcommand{\floor}[1]{\lfloor #1 \rfloor }
\begin{document}

\title{Wrapped flat-top kernel density estimation with circular data}

\author{Yasuhito Tsuruta\thanks{
School of Business Administration, Meiji University,
1-1 Kanda-Surugadai, Chiyoda-ku, Tokyo 101-8301, Japan.
Email: tsuruta\_y@meiji.ac.jp
}}

\date{}  % 日付を消す（任意）

\maketitle


\abstract {Kernel density estimators with circular data have been studied extensively for decades, as they allow flexible estimations even when the shape of the underlying density is complex. Many recent studies have examined bias correction methods; however, these methods are limited by the order when trying to improve the convergence rate of the bias, even if the true density is sufficiently smooth. To overcome this limitation,  the present study considers a new bias correction approach based on the characteristic functions of the underlying circular density. We introduce wrapped flat-top kernels, which are generated by wrapping the standard flat-top kernels defined on the real line onto the circumference of a unit circle. The asymptotic mean squared errors of the wrapped flat-top kernel density estimators are then derived. The results show that the convergence rate of these estimators is faster than that of previously introduced estimators. Furthermore, wrapped flat-top kernel density estimators achieve $\sqrt{n}$-consistency under the characteristic function of finite support, such as the circular uniform and cardioid distributions. We confirm these theoretical results in the numerical experiments. In empirical analyses, we also show that wrapped flat-top kernel density estimators effectively capture the shape of data. Therefore, such estimators are expected to allow flexible and accurate estimation in circular data analysis.}


%\pacs[MSC Classification]{62G07, 62H11}

\bigskip   
\section{Introduction}
If we consider circular data $\{\Theta_{1},\dots, \Theta_{n}\}$ generated by random sample from the underlying circular density $f(\theta)$ for $\theta \in [0,2\pi)$ with $\lim_{\theta \to 2\pi}f(\theta) \to f(0)$, typical examples are wind direction, paleocurrent direction, animal orientation, wildfire orientation, and the angles between certain atoms of amino acids in proteins \citetext{\citealp[Chapter 1]{mardia2000directional}; \citealp[Chapter 1]{Ley2017}}.

Kernel density estimators, which allow flexible estimations of such circular data, have been actively studied for approximately four decades \citetext{\citealp{beran1979exponential};  \citealp{hall1987kernel};  \citealp{bai1988kernel};  \citealp{taylor2008automatic}}. A circular kernel density estimator is defined as
\begin{align*}
\hat{f}_{h}(\theta) = \frac{1}{n}\sum_{i=1}^{n}K_{h}(\theta- \Theta_{i}),
\end{align*}
where $K_{h}(\theta)$ is a circular kernel function with bandwidth $h>0$. Throughout this manuscript, bandwidth $h$ is treated as a general smoothing
parameter that controls the concentration of the kernel function around its mean direction:
lower values of $h$ produce more concentrated kernel density estimators, while higher values yield
smoother estimators. In the real-valued setting, this coincides with the usual kernel function
rescaling $K_h(x)=h^{-1}K(x/h)$. In the circular setting, however, kernels are not
obtained from the literal rescaling of the argument; instead, $h$ controls the concentration of the kernel function through its functional form.

The integrated squared error (ISE) and mean integrated squared error (MISE) have often been employed as the global error criterion of kernel density estimators. These are defined as $\mathrm{ISE}[\hat{f}_{h}(\theta)] :=\int^{2\pi}_{0}[\hat{f}_{h}(\theta) -f(\theta)]^{2}d\theta$ and $\MISE[\hat{f}_{h}(\theta)]:= \E[\mathrm{ISE}[\hat{f}_{h}(\theta)]]$, respectively. By employing the minimizer $h_{*}$ for the MISE,  the convergence rate of the optimal MISE is $O(n^{-4/5})$ when we employ the non-negative circular kernel class proposed by \cite{ hall1987kernel}. Its rate when employing the wrapped Cauchy kernel is $O(n^{-2/3})$ \citep{tsuruta2017asymptotic, tenreiro2022kernel}. Furthermore, \cite{ameijeiras2024reliable} investigates the asymptotic MISE for wrapped kernels, such as a wrapped normal kernel.

Employing $p$th-order kernel functions, which do not satisfy non-negativity, reduces the bias of the MISE. Moreover, this approach ensures that the convergence rate of its optimal MISE is $O(n^{-2p/(2p+1)})$, under the assumption that the underlying circular density $f$ is $(p+2)$-th continuously differentiable \citep{tsuruta2017higher, tsuruta2024bias}. As \cite{politis1999multivariate} note, the convergence rate of the MISE is limited by the order of the kernel if the underlying density $f$ is sufficiently smooth. Although some non-negative bias correction methods have been proposed \citep{tsuruta2017higher,  bedouhene2022nonparametric,  tsuruta2024bias}, these estimators provide an MISE with the same order as the fourth-order kernel density estimators. Therefore, although these bias correction methods were inspired by the generalized jackknifing method for kernel density estimators on the real line proposed by \cite{jones1993generalized}, they have the same limitation as $p$th-order kernels.

We consider generating a kernel function achieving the order of the MISE as $O(n^{-2p/(2p+1)})$ desirable for any smoothness order $p$ of $f$. Thus, we must consider an infinity-order kernel function.
For data on the real line, many studies have discussed other bias correction methods based on the Fourier transform to generate an infinity-order kernel function \citetext{see \citealp{watson1963estimation}; \citealp{davis1977mean}; \citealp{devroye1992note}; \citealp{politis1999multivariate}}. \cite{devroye1992note} proposes a flat-top kernel function derived from the flat-top shape of the characteristic function. Employing a flat-top kernel function provides that the convergence rate of the mean squared error (MSE) is $O(n^{-1}\log n)$ when the characteristic function ofthe underlying density $f$ decays exponentially, and its rate is $O(n^{-1})$ when the characteristic function has finite support \citep{politis1999multivariate}. \cite{wang2022bias} argue that a transformed flat-top series estimator can reduce the bias of a density under compact support, although they focus on its local error.

While flat-top kernel functions on the real line are effective bias correction methods, few studies discuss flat-top kernel functions on the circle. \cite{politis1995bias} and \cite{politis2003adaptive} discuss flat-top kernel density estimators with the general trapezoidal shape window for a spectral density estimation. Furthermore, \cite{politis2001nonparametric} argues that a flat-top kernel function on the circle is obtained by wrapping it on the real line around the circumference of a unit circle. We refer to a kernel provided by using such a wrapping method as a wrapped flat-top kernel (see the next section for the definition). This study aims to investigate the asymptotic properties (e.g., the MISE) and small sample behaviors of selected wrapped-flat top kernel functions. 

The remainder of this paper is organized as follows. Section \ref{sec: theories} discusses the asymptotic properties of wrapped flat-top kernel density estimators. Section \ref{sec: numerical} investigates the small sample characteristics of wrapped flat-top kernel density estimators and compares these estimators and the previously introduced estimators in a numerical experiment.  Section \ref{sec: application} presents the result of empirical analyses using the wrapped flat-top kernel density estimators. Finally, Section \ref{sec: conclusion} concludes.  The proofs of the theorems are in the Appendix.
\section{Properties of wrapped flat-top kernel functions}\label{sec: theories}
We denote the characteristic function of circular density $f$ as $\phi_{t}(f): = \int^{2\pi}_{0}f(\theta)e^{it\theta}d\theta$, where the complex conjugate of $\phi_{t}(f)$ with $\bar{\phi}_{t}(f) = \phi_{-t}(f).$ We assume that any circular kernel function satisfies $\int^{\pi}_{-\pi}K_{h}(\theta) d\theta= 1$ with the following Fourier series
\begin{align*}
K_{h}(\theta) = \frac{1}{2\pi}\sum_{t\in \mathbb{Z}}\phi_{t}(K_{h})e^{-it\theta},%\label{eq: KF},
\end{align*}
where $\mathbb{Z}$ is the set of integers. Therefore, a kernel density estimator can be rewritten as
\begin{align*}
\hat{f}_{h}(\theta) = \frac{1}{n}\sum_{j=1}^{n}\sum_{t\in \mathbb{Z}}\phi_{t}(K_{h})e^{-it\theta}e^{it\Theta_{j}}.
\end{align*}
Any integrable function on the circle satisfies Parseval's identity \citep[p. 80]{ stein2011fourier}.  Therefore, we obtain the exact MISE expression as follows:
\begin{theorem}\label{theo: MISE}
For circular density $f$, it holds that 
\begin{align}
\MISE[\hat{f}_{h}(\theta)]&= \ISB[\hat{f}_{h}(\theta)] +  \IV[\hat{f}_{h}(\theta)],\notag
\shortintertext{where}
 \ISB[\hat{f}_{h}(\theta)] &=\frac{1}{2\pi}\sum_{t \in \mathbb{Z}}|\phi_{t}(f)|^{2}|1 - \phi_{t}(K_{h})|^{2},\notag
\shortintertext{and}
  \IV[\hat{f}_{h}(\theta)]&=  \frac{1}{2\pi n}\sum_{t \in \mathbb{Z}}|\phi_{t}(K_{h})|^{2}(1 - |\phi_{t}(f)|^{2}), \notag %\label{eq: MISE}
\end{align}
are called the integrated square bias (ISB) and integrated variance (IV), respectively.
\end{theorem}
A flat-top kernel function $K_{h,c}(x)$ is defined as a kernel function on the real line with the characteristic function $\phi(t; K_{h,c}):= \int_{x\in \mathbb{R}}K_{h,c}(x)e^{itx}dx$ for $t \in \mathbb{R}$, such that
\begin{align*}
\phi(t; K_{h,c}) := 
\begin{cases}
1 & \text{if}\quad |t|\le 1/h,\\
g(t; h^{-1},c) & \text{if}\quad  1/h < |t|< c/h,\\
0 & \text{if}\quad  |t|\ge c/h,
\end{cases}
\end{align*}
where where $g(t; h^{-1},c)$ is the real-valued function, $0 < |g(t; h^{-1},c)| < 1$, and $g(t; h^{-1},c) = g(-t; h^{-1},c)$ with any real number $c\ge1$.  the function $g(t; h^{-1},c)$ is usually continuous \citep{politis1999multivariate}. 

We introduce the floor function $\floor{x} :=\max \{ z \in \mathbb{Z}\mid z\le x\}$.
We can wrap a flat-top kernel function on the real line with bandwidth $1/\floor{\nu}$ around the circumference of a unit circle. This motivates the definition of a wrapped flat-top kernel function as follows.
\begin{definition}
Let $c$ be any integer with $c\ge1$ and  $\nu \in [0, \infty)$ be a smoothing parameter. Futhermore,  the real-valued function $g(t; \nu,c)$ satisfies $0 < |g(t;\nu,c)| <1$ and $g(t; h^{-1},c) = g(-t; h^{-1},c)$
We define the wrapped flat-top kernel function $K_{\nu,c}$ as a kernel function whose characteristic function is given by
\begin{align}
\phi_{t}(K_{\nu,c}) := 
\begin{cases}
1 & \text{if}\quad |t|\le \floor{\nu},\\
g(t; \nu ,c) & \text{if}\quad  \floor{\nu} < |t|< c\floor{\nu},\\
0 & \text{if}\quad  |t|\ge c\floor{\nu},
\end{cases}\label{eq; chrWFTK}
\end{align}
if $\nu > 0$. If $\nu=0$, its characteristic function is given by
\begin{align}
\phi_{t}(K_{0,c}) := 
\begin{cases}
1 & \text{if}\quad |t|= 0,\\
0 & \text{if}\quad |t| >0.
\end{cases}\label{eq; chrWFTK}
\end{align}
\end{definition}
To guarantee the differentiability of $\MISE[\hat{f}_{h}(\theta)]$ at $\nu$, the support of $\nu$ must be the positive real line. In practice, statisticians may choose any smoothing parameter from a set of non-negative integers, as different choices of $\nu$ yield different kernel density estimators.

We obtain a wrapped flat-top kernel function with $\phi_{t}(K_{\nu,c}) = \phi(t; K_{ 1/\floor{\nu},c})$, where $\phi(t; K_{ 1/\floor{\nu},c})$ is the characteristic function of any flat-top kernel function on the real line with $\floor{\nu} >0$. For instance, we obtain the wrapped trapezoid kernel function from the trapezoid kernel function, which is given by
\begin{align*}
K_{\nu, c}^{\wtk}(\theta) &:= \begin{dcases}
\frac{1}{2\pi}& \text{if} \quad \nu =0, \quad \\
 \frac{\sin^{2}(c\floor{\nu} \theta/2) - \sin^{2}(\floor{\nu} \theta/2)}{2\pi (c-1)\floor{\nu}\sin^{2}(\theta/2)}& \text{if} \quad \nu >0,
\end{dcases}
\shortintertext{with}
g^{\wtk}(t; \nu ,c) &:= \frac{c- |t|/\floor{\nu}}{c-1}.
\end{align*}
For $\nu >0$, $ K_{\nu, c}^{\wtk}(0) = (2\pi)^{-1} (c+1)\floor{\nu}$.
The wrapped trapezoid kernel function is also obtained by the reparametrization of the de la vall\'{e}e Poussin kernel \citep{mehta2015l1}.  When $c = 1$, we obtain the wrapped sinc kernel function as follows:
\begin{align*}
K_{\nu, 1}^{\wsk}(\theta) &:=
 \frac{\sin((\floor{\nu}+1/2)\theta)}{2\pi \sin(\theta/2)}
\end{align*}
where $\phi_{t}(K_{\nu,1}^{\wsk}) =1$ if $|t|\le \floor{\nu}$, and 0 otherwise. Notably, $K_{\nu, 1}^{\wsk}(0) = (2\pi)^{-1} (2\floor{\nu} +1)$. In  Fourier analysis, the wrapped sinc kernel function is often called the Dirichlet kernel. Figure \ref{fig: wfk} plots these kernel functions.
We denote the wrapped flat-top kernel density estimator as $\hat{f}_{\nu,c}(\theta):= n^{-1}\sum_{i}K_{\nu,c}(\theta - \Theta_{i})$.  

\begin{figure}[htbp]
    \begin{tabular}{cc}
        \begin{minipage}[t]{0.40\hsize}
        \centering
        \includegraphics[keepaspectratio, scale=0.38]{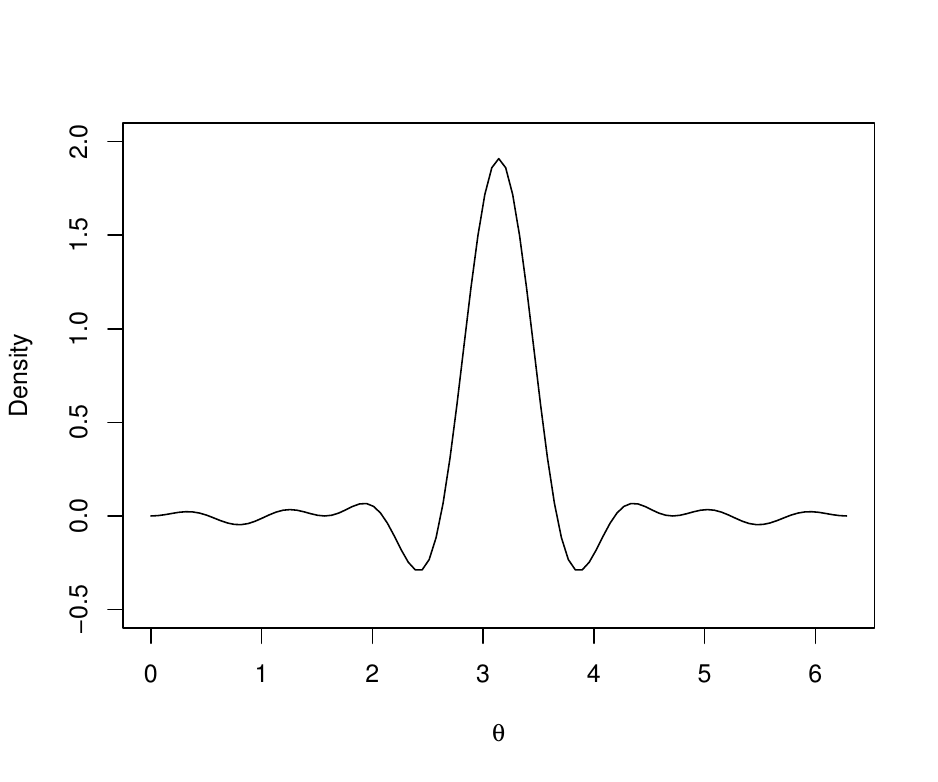}
        \subcaption{Wrapped trapezoid kernel function with $\nu=4$, $c=2$, and mean direction $\mu=\pi$.}
        \label{fig: wtk}
      \end{minipage}&
      \begin{minipage}[t]{0.40\hsize}
        \centering
        \includegraphics[keepaspectratio, scale=0.38]{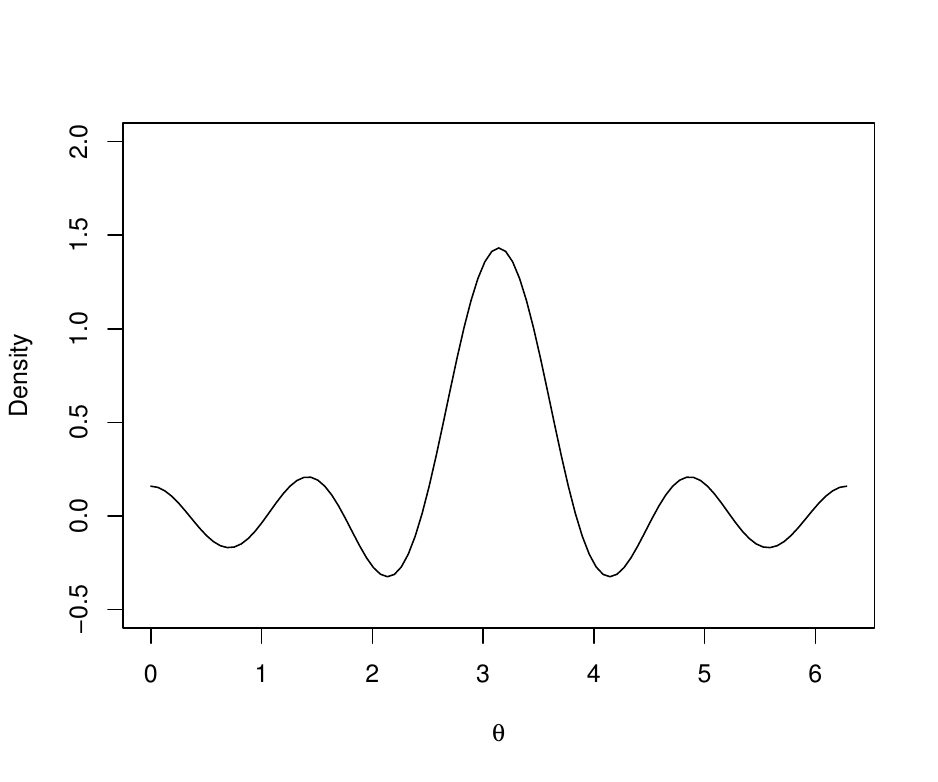}
        \subcaption{Wrapped sinc kernel function with $\nu=4$ and mean direction $\mu=\pi$.}
        \label{fig: wsk}
      \end{minipage} \\
        \begin{minipage}[t]{0.40\hsize}
        \centering
        \includegraphics[keepaspectratio, scale=0.38]{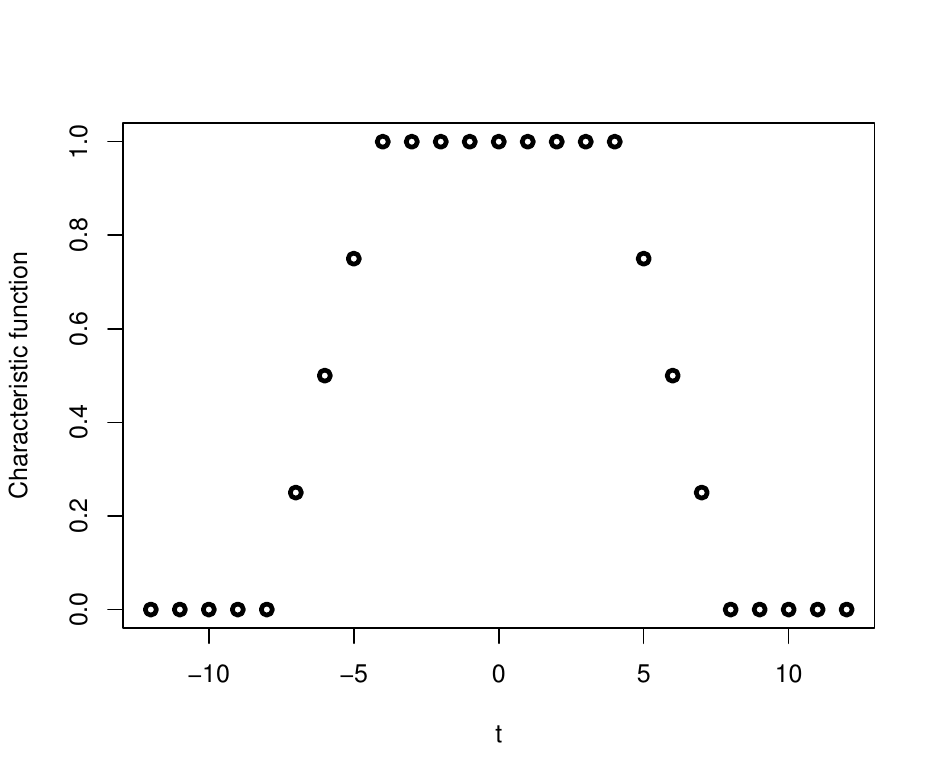}
        \subcaption{Characteristic function of wrapped trapezoid kernel with $\nu=4$ and $c=2$.}
        \label{fig: c_wtk}
      \end{minipage} &
      \begin{minipage}[t]{0.40\hsize}
        \centering
        \includegraphics[keepaspectratio, scale=0.38]{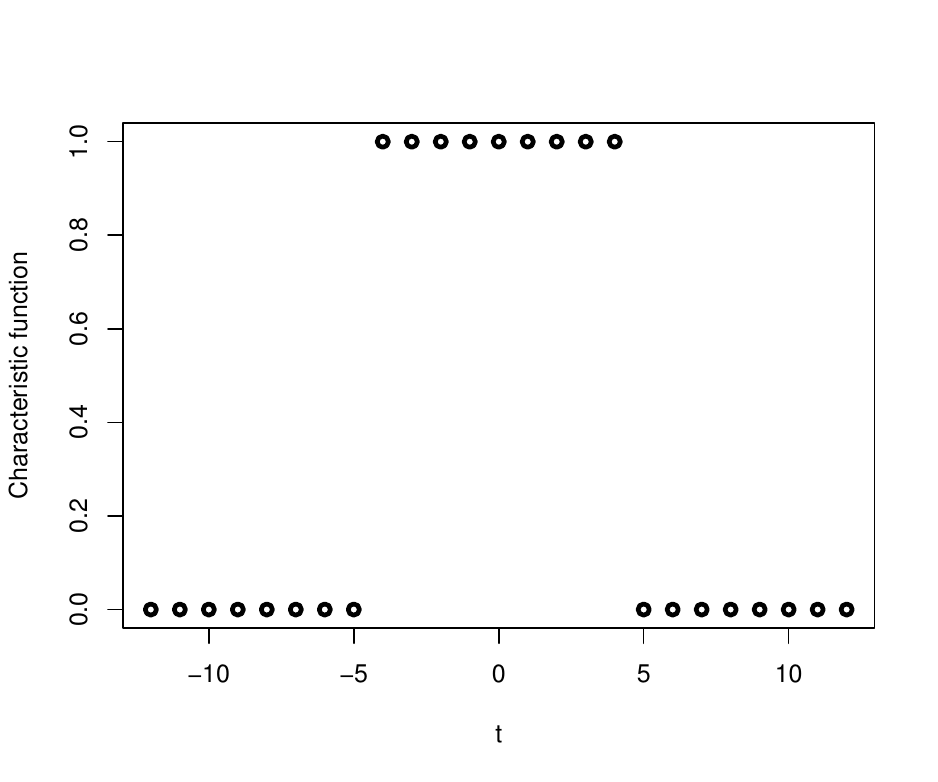}
        \subcaption{Characteristic function of wrapped sinc kernel function with $\nu=4$.}
        \label{fig: c_wsk}
      \end{minipage}
    \end{tabular}
     \caption{Wrapped flat-top kernel functions.}
\label{fig: wfk}
  \end{figure}
  
Wrapped flat-top kernel density estimators have the following exact MISE:
\begin{align}
\ISB[\hat{f}_{\nu,c}(\theta)] &=\frac{1}{2\pi}\left[\sum_{ \floor{\nu} < |t| < c \floor{\nu}}|\phi_{t}(f)|^{2}|1 - \phi_{t}(K_{\nu,c})|^{2}+ \sum_{c \floor{\nu} \le |t|}|\phi_{t}(f)|^{2}\right] ,\label{eq: ISB0}
\shortintertext{and}
\IV[\hat{f}_{\nu,c}(\theta)]&=\frac{1}{2\pi n}\sum_{ |t| < c \floor{\nu}}|\phi_{t}(K_{\nu,c})|^{2}(1 - |\phi_{t}(f)|^{2}).\label{eq: IV0}
\end{align}
We obtain the order of the IV for a wrapped flat-top kernel density estimator as follows.
\begin{theorem}\label{theo: IVwftk}
When we employ a wrapped flat-top kernel function, for circular density $f$, it holds that
\begin{align*}
\IV[\hat{f}_{\nu,c}(\theta)]\le \frac{c\nu}{\pi n}.
\end{align*}
\end{theorem}
The convergence rate of the supremum of the MSE for wrapped flat-top kernel density estimators was derived by \cite{politis2001nonparametric}.  However, his result does not provide the constant associated with the supremum of the MSE. Furthermore, the literature discussing whether circular densities satisfy the assumptions laid out in his work remains sparse. 
 
To derive the ISB employing its estimator, we introduce the following three smoothness conditions applicable to commonly used circular densities:
\begin{enumerate}[\textrm{C}1)]
\item There exists $r \in (0, \infty)$ such that
\begin{align*}
\frac{1}{2\pi}\sum_{t \in \mathbb{Z}}t^{2r}|\phi_{t}(f)|^{2}=:C_{r}(f) < \infty.
\end{align*}
% Density $f(\theta)$ is $p$ times continuously differentiable.
\item There exists positive constants $\alpha$ and $\tau$ such that
\begin{align*}
\frac{1}{2\pi}\sum_{t \in \mathbb{Z}}e^{\tau|t|^{\alpha}}|\phi_{t}(f)|^{2}=:I_{\alpha, \tau}(f) < \infty.
\end{align*}
\item There exists a non-negative integer $T_{f}$ such that $\phi_{t}(f)=0$ for $|t|>T_{f}$.
\end{enumerate}

The characteristic function converges faster as the smoothness of $f$ increases \citep[pp. 42-43]{ stein2011fourier}. The standard assumption in circular kernel density estimation studies is that circular density $f(\theta)$ is $p$ times continuously differentiable. When $r=p$, Parseval's identity implies $C_{p}(f)=R(f^{(p)}) < \infty$, where $R(g) := \int^{2\pi}_{0}g(\theta)^{2}d\theta$. Thus, Condition C1) corresponds to assuming that the $p$-th derivative of $f$ is square-integrable. Moreover, Condition C3) implies Condition C2), and Condition C2) implies Condition C1).

Conditions C1), C2), and C3) are essentially the same as those introduced in  \cite{politis2001nonparametric}, namely,  $K_{1}$, $K_{2}$, and $K_{3}$ . Conditions C1) and $K_{1}$ concern the polynomial decay of the characteristic function $\phi_{t}$, while Conditions C2) and $K_{2}$ concern its exponential decay. Condition C1) is suitable for expressing the roughness magnitude $R(f^{(p)})$, whereas Condition C2) can be used to derive the constant of the upper bound of the MISE. Condition C3) is the same as condition $K_{3}$.
\begin{theorem}\label{theo: rMISE}
We assume that Condition C1) holds. Then, employing a wrapped flat-top kernel function, we obtain
\begin{align}
\ISB[\hat{f}_{\nu,c}(\theta)] &\le \frac{1}{v^{2r}}C_{r}(f). \label{eq: ISB2r}
\end{align}
Combining Theorem \ref{theo: IVwftk} and \eqref{eq:  ISB2r} yields
\begin{align}
\MISE[\hat{f}_{\nu,c}(\theta)]&\le \frac{1}{\nu^{2r}}C_{r}(f) + \frac{c\nu}{\pi n}.\label{eq: MISE2r}
\end{align}
The minimizer of the right-hand side of \eqref{eq: MISE2r} is 
\begin{align*}
\nu_{r} &= \left[\frac{2\pi r C_{r}(f)n}{c}\right]^{1/(2r+1)}.
\end{align*}
When employing $\nu_{r}$, the order of the minimizing MISE is $O(n^{-2r/(2r+1)})$.
\end{theorem}
\begin{remark}\label{rem: rMISE}
Condition C1) involves non-differentiable densities such as a triangular density, which is defined as
\begin{align*}
f_{\rho}^{\mathrm{TR}}(\theta):= \frac{1}{8\pi}\{4 - \pi^{2}\rho + 2\pi\rho|\pi - \theta|\}, \quad \rho \in \left[0, \frac{4}{\pi^{2}}\right], 
\end{align*}
and its characteristic functions are given by $\phi_{2t-1}(f_{\rho}^{\mathrm{TR}})= \rho/(2t-1)^2$ with all the even moments being zero \citep[p. 34]{Jammalamadaka2001Topics}. Its density shape is a triangle and it is symmetric about zero. For the triangular density, we can use the positive value $r = 3/2- \epsilon/2$ with a low value $\epsilon>0$ to satisfy Condition C1). Then, we find that the order of the minimizing MISE is $O(n^{-(3-\epsilon)/4-\epsilon})$. Previously introduced kernel density estimators cannot derive the asymptotic MISE for non-differentiable densities. For instance, to calculate the asymptotic MISE, a rose diagram, which is a histogram-type estimator on the circle and not smooth, must assume that the underlying circular density is differentiable \citep{tsuruta2023automatic}. This indicates that wrapped flat-top kernel density estimators are useful for more densities than the previously introduced estimators.
\end{remark}

\begin{theorem}\label{theo: expMISE}
We assume that Condition C2) holds. Then, employing a wrapped flat-top kernel function, we obtain
\begin{align}
\ISB[\hat{f}_{\nu,c}(\theta)] \le I_{\alpha, \tau} e^{-\tau|\nu|^{\alpha}}.\label{eq: ISBexp}
\end{align}
Employing $\nu = O((\tau^{-1}\log n)^{1/\alpha})$ yields $\mathrm{ISB}[\hat{f}_{\nu}(\theta)]=O(n^{-1})$. This and Theorem \ref{theo: IVwftk} imply
\begin{align}
\MISE[\hat{f}_{\nu,c}(\theta)] \le  \frac{c\tau^{-1} (\log n)^{1/\alpha}}{\pi n} + O(n^{-1})
\end{align}
and the order of the minimizing MISE is $O((\log n)^{1/\alpha}n^{-1})$.
\end{theorem}
\begin{remark}\label{rem: expMISE}
A typical example of Condition C2) is the wrapped $\alpha$-stable distribution $S_{\alpha}(\tau, \beta, \mu)$, where $\alpha \in (0, 2]$, $\beta \in [-1,1]$, $\tau \in (0, \infty)$, and $\mu \in [0, 2\pi)$ are the parameters of stability, skewness, concentration, and mean direction, respectively \citep{Gatto2003inference}. When $\beta =0$, the subclass  $S_{\alpha}(\tau, 0 , \mu)$ is symmetric about $\mu$. The characteristic function of $S_{\alpha}(\tau, \beta, \mu)$ is given by
\begin{align*}
\phi_{t}(f^{W\alpha S}_{\alpha, \tau, \beta, \mu}):= \begin{cases}
\exp\{-\tau^{\alpha}|t|^{\alpha}[1- i\beta\sgn(t)\tan(\alpha\pi/2)] + i\mu t\}&\text{if}\quad \alpha \in (0,1) \cup (1,2],\\
\exp(-\tau|t| + i\mu t)&\text{if}\quad \alpha = 1.
\end{cases}
\end{align*}
Two well-known examples of the symmetric wrapped  $\alpha$-stable distribution are the wrapped normal distribution with $\alpha = 2$ ($\beta$ is irrelevant) and the wrapped Cauchy distribution with  $\alpha = 1$ and $\beta=0$. The wrapped normal distribution can approximate the von Mises distribution with a sufficiently large concentration parameter \citep[p. 38]{mardia2000directional}.
\end{remark}
\begin{theorem}\label{theo: boundMISE}
We assume that Condition C3) holds. Then, employing a wrapped flat-top kernel function that satisfies $\nu \in [T_{f}, B]$, we obtain
\begin{align}
\ISB[\hat{f}_{\nu,c}(\theta)] = 0. \label{eq: ISBTf}
\end{align}
Combining Theorem \ref{theo: IVwftk} and \eqref{eq: ISBTf} yields
\begin{align}
\MISE[\hat{f}_{v,c}(\theta)]\le \frac{c\nu}{\pi n}.\label{eq: MISE2p}
\end{align}
The right-hand side of \eqref{eq: MISE2p} with $\nu = T_{f}$ is the smallest for $\nu \in [T_{f}, B]$. 
\end{theorem}
We obtain the proof by combining Theorems \ref{theo: MISE} and \ref{theo: IVwftk}  for $\nu \in [T_{f}, B]$. 
\begin{remark}\label{rem: boundMISE}
Theorem \ref{theo: boundMISE} implies that a wrapped flat-top kernel allows one to achieve $\sqrt{n}$-consistency. This convergence rate was already established by \cite{politis2001nonparametric}. The contribution of Theorem \ref{theo: boundMISE} lies in identifying the constant of the upper bound of the MISE.
Some distributions satisfy Condition C3). The example with $T_{f}=0$ is the circular uniform distribution with the following characteristic function: the value is one if $t= 0$, and zero otherwise. 
Next, a typical example of a distribution with $T_{f}=1$ is the cardioid distribution whose density is
\begin{align}
f^{C}_{\rho, \mu}(\theta)= \frac{1}{2\pi}[1+2\rho\cos(\theta-\mu)],\quad \rho \in \left[0,  \frac{1}{2}\right],
\end{align}
where $\rho$ is the concentration parameter and its characteristic function is 
$\phi_{1}(f^{C}_{\rho,\mu})= \rho e^{i\mu}$ and $\phi_{t}(f^{C}_{\rho,\mu})=0$ for $t\ge2$.  
\end{remark}
Hence, we establish the theoretical convergence rate of the MISE for these popular circular densities. Under Conditions C1) or C3), the convergence rate is similar to that of the supremum of the MSE derived by \cite{politis2001nonparametric}. However, under Condition C2), the convergence rate differ slightly.
\section{Numerical experiments}\label{sec: numerical}
We conduct two numerical experiments to analyze the finite sample behavior of the wrapped flat-top kernel density estimators. The simulations are conducted using the \textsf{R} statistical software. Section \ref{subsec: optimal} shows the small sample characteristics of the wrapped flat-top kernel density estimators, with the theoretical smoothing parameter derived in Section \ref{sec: theories}. In Section \ref{sec: selectors}, we propose smoothing parameter selectors for the wrapped flat-top kernel density estimators. Section \ref{subsec: compare} compares the wrapped flat-top kernel with the previously introduced estimators using selected smoothing parameter selectors.

\subsection{Finite sample behavior of the wrapped flat-top kernel estimators}\label{subsec: optimal}
We use four of the 20 standard scenarios introduced by \cite{oliveira2012plug}, as shown in Table \ref{tab; scenario} (see \cite{oliveira2012plug} for the plots): the circular uniform distribution (M1), wrapped normal distribution (M3), cardioid distribution (M4), and wrapped Cauchy distribution (M5).  These random samples are generated in \texttt{dcircmix} function of the \texttt{NPCirc} package \citep{oliveira2014npcirc}.  The performance measure is the MISE based on 1000 repetitions, with the sample sizes $n=50, 100, 200, 400, 800$, and $1600$. 

For the wrapped trapezoid kernel, we employ $c=2$, as recommended by \cite{politis1999multivariate}.
In each scenario, we employ the optimal smoothing parameter. We select $\nu=0$ and $\nu=1$ for the circular uniform and cardioid distributions, respectively, with these scenarios included in Condition C3). Furthermore, we select $\nu= \floor{(\tau^{-1}\log n)^{1/\alpha}}$ for the wrapped normal and wrapped Cauchy distributions, which are $\tau= -\log 0.9$ and $\alpha = 2$ in the former and $\tau= -\log 0.8$ and $\alpha = 1$ in the latter. These scenarios are included in Condition C2).

As shown in Figure \ref{fig: log-log}, when employing the optimal smoothing parameter, the MISEs of the wrapped sinc and trapezoid kernel density estimators are almost the same as the theoretical convergence rate of the MISE derived in Section \ref{sec: theories} in all the scenarios. However, the logarithms of M3 and M5 are curves as opposed to strictly straight lines, as these have the terms $-2^{-1}\log\log{n}$ or $-\log\log{n}$, respectively, which exhibit little change because $\log\log{50} \simeq 1.364$ and $\log\log{1600}  \simeq 1.998$. Therefore, the wrapped flat-top kernel density estimators with the theoretical optimal smoothing parameter have the same order as the theoretical MISE. Hence, we confirm these theoretical performances. Furthermore, in the scenarios included in  Condition C3) (M1 and M4), the MISEs of the wrapped flat-top kernel density estimators are almost the same. In the scenarios included in Condition C2) (M3 and M5), the MISE of the wrapped sinc kernel density estimator is lower than that of the wrapped trapezoid kernel density estimator.
\begin{table}[htbp]
  \centering
  \caption{Scenarios introduced by \cite{oliveira2012plug}. $\mathrm{vM}(\mu, \kappa)$ denotes the von Mises distribution with mean direction $\mu$ and concentration parameter $\kappa$. $\mathrm{WN}(\mu, \rho)$, $\mathrm{C}(\mu, \rho)$, and $\mathrm{WC}(\mu, \rho)$ denote the wrapped normal, cardioid, and wrapped Cauchy distributions, with concentration parameter $\rho$. $\mathrm{WSN}(\xi, \eta, \lambda)$ denotes the wrapped skew-normal distribution characterized by location parameter $\xi$, scale parameter $\eta$, and skewness
parameter $\lambda$ \citep{pewsey2000wrapped}.}
    \begin{tabular*}{\textwidth}{@{\extracolsep\fill}cp{12cm}}
\toprule
    \multicolumn{1}{l}{No.} & \multicolumn{1}{l}{Distribution } \\\midrule
M1 & Circular uniform\\
M2 &  Von Mises $\mathrm{vM}(\pi, 1)$\\
M3 & Wrapped normal $\mathrm{WN}(\pi, 0.9)$\\
M4 & Cardioid $\mathrm{C}(\pi, 0.5)$\\
M5 & Wrapped Cauchy $\mathrm{WC}(\pi, 0.8)$\\
M6 & Wrapped skew--normal $\mathrm{WSN}(\pi, 1, 20)$\\
M7 & Mixture of two von Mises $1/2\mathrm{vM}(0,4) + 1/2\mathrm{vM}(\pi,4) $\\
M8 & Mixture of two von Mises $1/2\mathrm{vM}(2,5) + 1/2\mathrm{vM}(4,5) $\\
M9 & Mixture of two von Mises $1/4\mathrm{vM}(0,2) + 3/4\mathrm{vM}(\pi/\sqrt{3},2) $\\
M10 & Mixture of von Mises and wrapped Cauchy:\par $4/5\mathrm{vM}(\pi,5) + 1/5\mathrm{WC}(4\pi/3,0.9) $\\
M11 & Mixture of three von Mises:\par $1/3\mathrm{vM}(\pi/3,6) + 1/3\mathrm{vM}(\pi,6) +  1/3\mathrm{vM}(5\pi/3,6) $\\
M12 & Mixture of three von Mises $2/5\mathrm{vM}(\pi/2,4) + 1/5\mathrm{vM}(\pi,5) + 2/5\mathrm{vM}(3\pi/2,4) $\\
M13 & Mixture of three von Mises:\par $2/5\mathrm{vM}(0.5,6) + 2/5\mathrm{vM}(3,6) + 1/5\mathrm{vM}(5,24) $\\
M14 & Mixture of four von Mises:\par$1/4\mathrm{vM}(0,12) + 1/4\mathrm{vM}(\pi/2,12) + 1/4\mathrm{vM}(\pi,12) + 1/4\mathrm{vM}(3\pi/2,12) $\\
M15 & Mixture of wrapped Cauchy, wrapped normal, von Mises and wrapped skew-normal:\par $3/10\mathrm{WC}(\pi-1,0.6) + 1/4\mathrm{WN}(\pi+0.5,0.9) + 1/4\mathrm{vM}(\pi+2,3) + 1/5\mathrm{WSN}(6,1,3) $\\
M16& Mixture of five von Mises:\par$1/5\mathrm{vM}(\pi/5,18) + 1/5\mathrm{vM}(3\pi/5,18) + 1/5\mathrm{vM}(\pi,18) + 1/5\mathrm{vM}(7\pi/5,18) + 1/5\mathrm{vM}(9\pi/5,18)$\\
M17 & Mixture of cardioid and wrapped Cauchy $2/3\mathrm{C}(\pi,0.5) +  1/3\mathrm{WC}(\pi,0.9)$\\
M18 & Mixture of four von Mises:\par $1/2\mathrm{vM}(\pi,1) + 1/6\mathrm{vM}(\pi-0.8,30) + 1/6\mathrm{vM}(\pi,30) + 1/6\mathrm{vM}(\pi + 0.8,30) $\\
M19 & Mixture of five von Mises:\par $4/9\mathrm{vM}(2,3) + 5/36\mathrm{vM}(4,3) + 5/36\mathrm{vM}(3.5,50) + 5/36\mathrm{vM}(4, 50) + 5/36\mathrm{vM}(4.5, 50) $\\
M20 & Mixture of two wrapped skew-normal and two wrapped Cauchy:\par $1/3\mathrm{WSN}(0, 0.7, 20) + 1/3\mathrm{WSN}(\pi,0.7,20) + 1/6\mathrm{vM}(3\pi/4, 0.9) + 1/6\mathrm{vM}(7\pi/4, 0.9) $\\
\bottomrule
    \end{tabular*}%
  \label{tab; scenario}%
\end{table}%

\subsection{Smoothing parameter selectors}\label{sec: selectors}
We introduce two smoothing parameter selectors for the wrapped kernel density estimators: the least squares cross-validation (LSCV) and empirical rule (ER) selectors. The ER selector is proposed by \cite{politis2003adaptive}.  These selectors are nonparametric estimators that do not depend on any parametric density. The LSCV selector is defined as follows.
\begin{definition}
Denote that $\hat{f}_{\nu, c, -i}(\theta):= (n-1)^{-1}\sum_{j\neq i}K_{\nu,c}(\theta -\Theta_{j})$. Let $L$ be any positive integer. Then,  the LSCV selector is given by
\begin{align*}
\hat{\nu}_{\cv}&:= \argmin_{\nu\in\{0,1,2,\dots, L\}} \cv(\nu),
\shortintertext{where}
 \cv(\nu) &:= \int^{2\pi}_{0}\hat{f}_{\nu,c}(\theta)^{2}d\theta -\frac{2}{n}\sum_{i=1}^{n}\hat{f}_{\nu,c,-i}(\Theta_{i}).
\end{align*} 
\end{definition}

\cite{politis2003adaptive} defines the ER selector for picking $\nu$ as follows.
\begin{definition}
Define the empirical characteristic function for any density $f$ as $\hat{\phi}_{t}(f) :=n^{-1}\sum_{i=1}^{n}e^{-it\Theta_{i}}$. Let $M$ be a fixed positive constant and $\ell_{n}$ be a positive non-decreasing sequence.  Then, the ER selector is given by
\begin{align*}
\hat{\nu}_{\er} := \inf\left\{\nu \in \{0, 1, 2,\dots\}\mid |\hat{\phi}_{\nu + t}(f)|^{2} < M \frac{\log n }{n},\quad \forall t \in \{1,2,\dots, \ell_{n}\}\right\}.
\end{align*}
\end{definition}
\cite{politis2003adaptive} mentions the asymptotic properties of the ER selector for spectral density functions. \cite{chacon2007note} employ its selector to estimate $D_{f}:=\sup\{t\ge0: \phi(t;f) \neq 0\}$.  
\cite{politis2003adaptive} recommends $M = 1$, $\ell_{n} = 1$ or  $M = 4$, $\ell_{n} = 5$ as selecting the parameters of the ER selector for data on the real line. For the numerical experiment, $M = 1$, $\ell_{n} = 5$ is chosen  heuristically. 
\begin{figure}[htbp]
    \begin{tabular}{cc}
      \begin{minipage}[t]{0.45\hsize}
        \centering
        \includegraphics[keepaspectratio, scale=0.15]{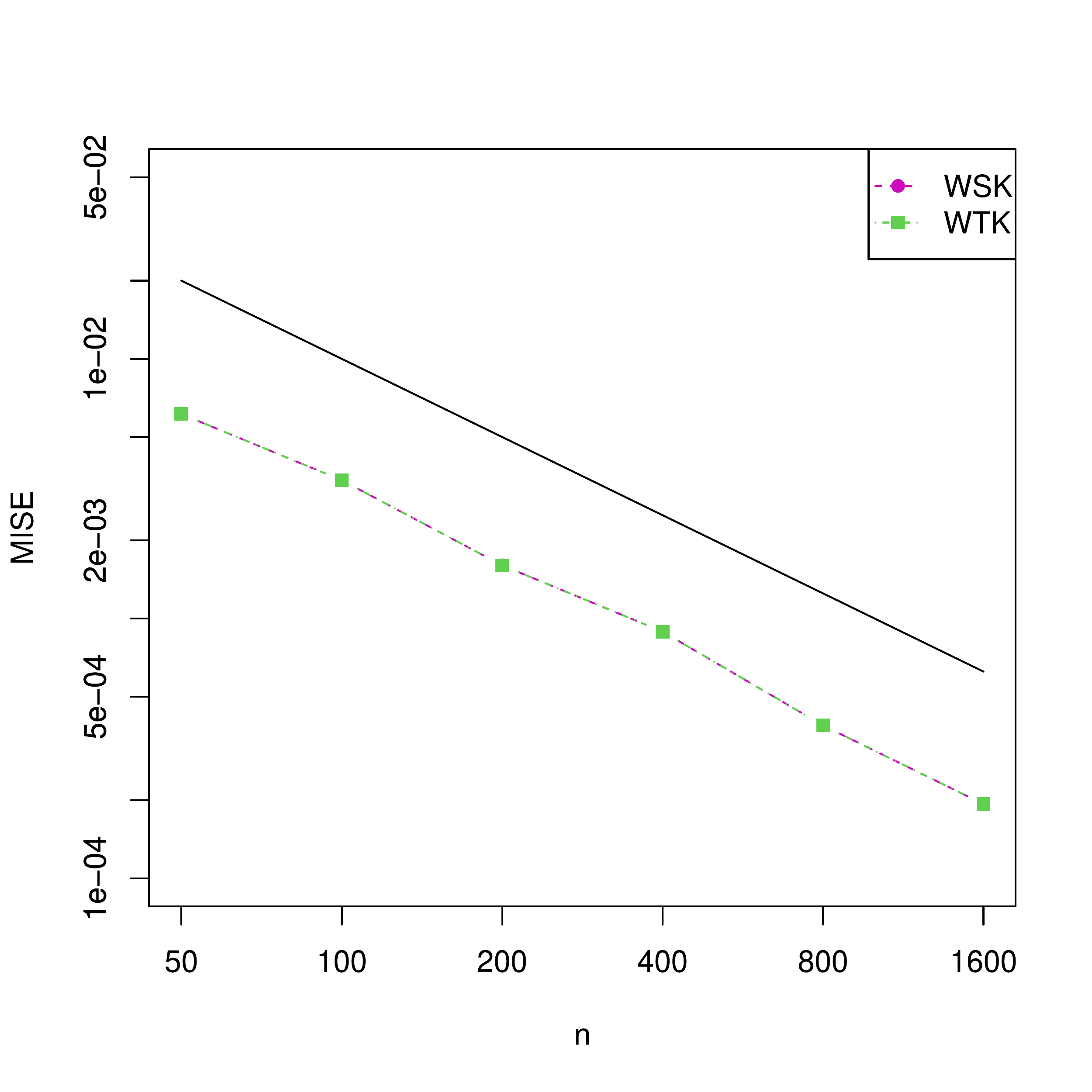}
        \subcaption{Circular uniform distribution (M1). The solid line represents $n^{-1}$, which is the order of the MISE provided in Theorem \ref{theo: boundMISE}.}
        \label{fig: cu}
      \end{minipage} &
      \begin{minipage}[t]{0.45\hsize}
        \centering
        \includegraphics[keepaspectratio, scale=0.15]{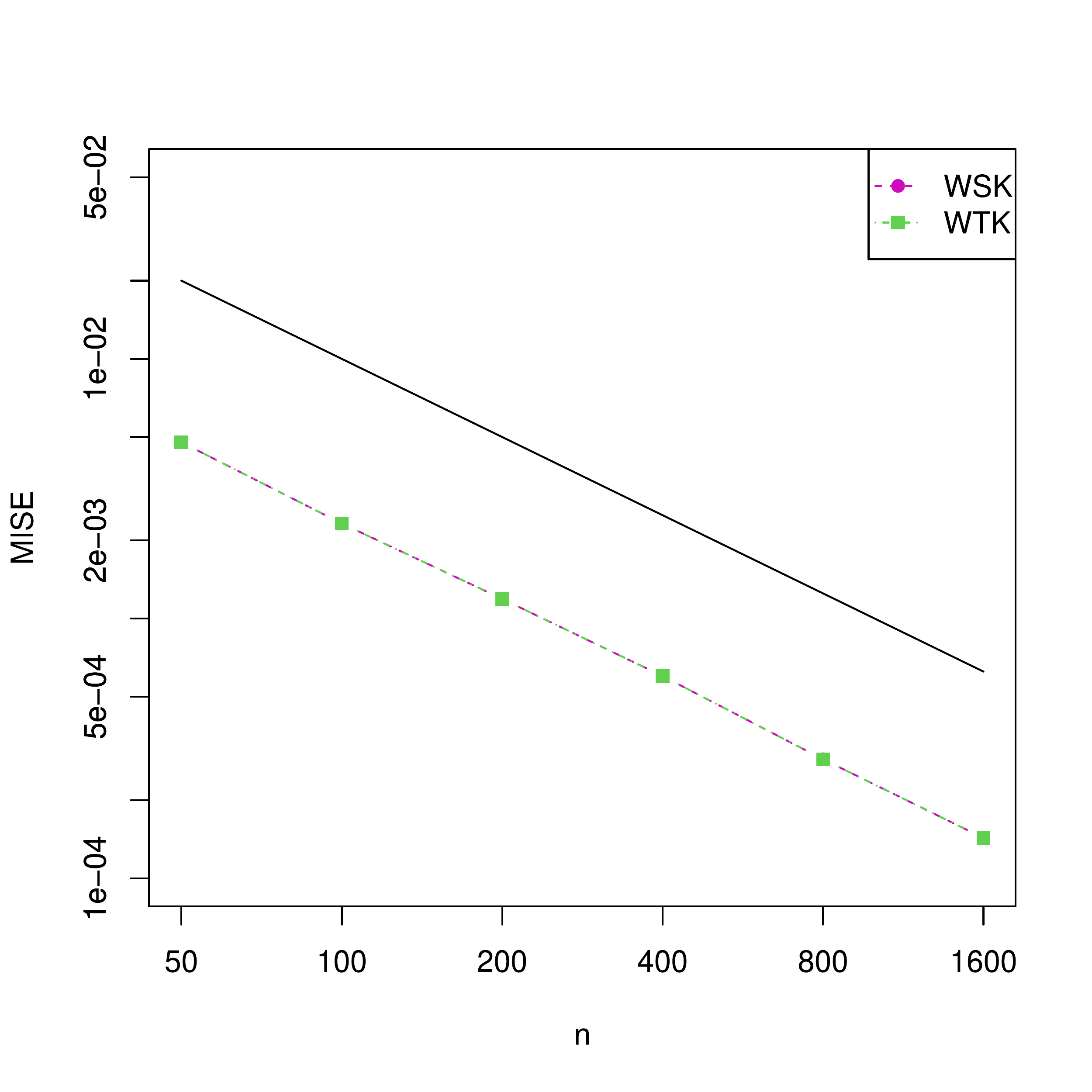}
        \subcaption{Cardioid distribution (M4). The solid line represents $n^{-1}$, which is the order of the MISE provided in Theorem \ref{theo: boundMISE}.}
        \label{fig: ca}
      \end{minipage} \\
   
      \begin{minipage}[t]{0.45\hsize}
        \centering
        \includegraphics[keepaspectratio, scale=0.15]{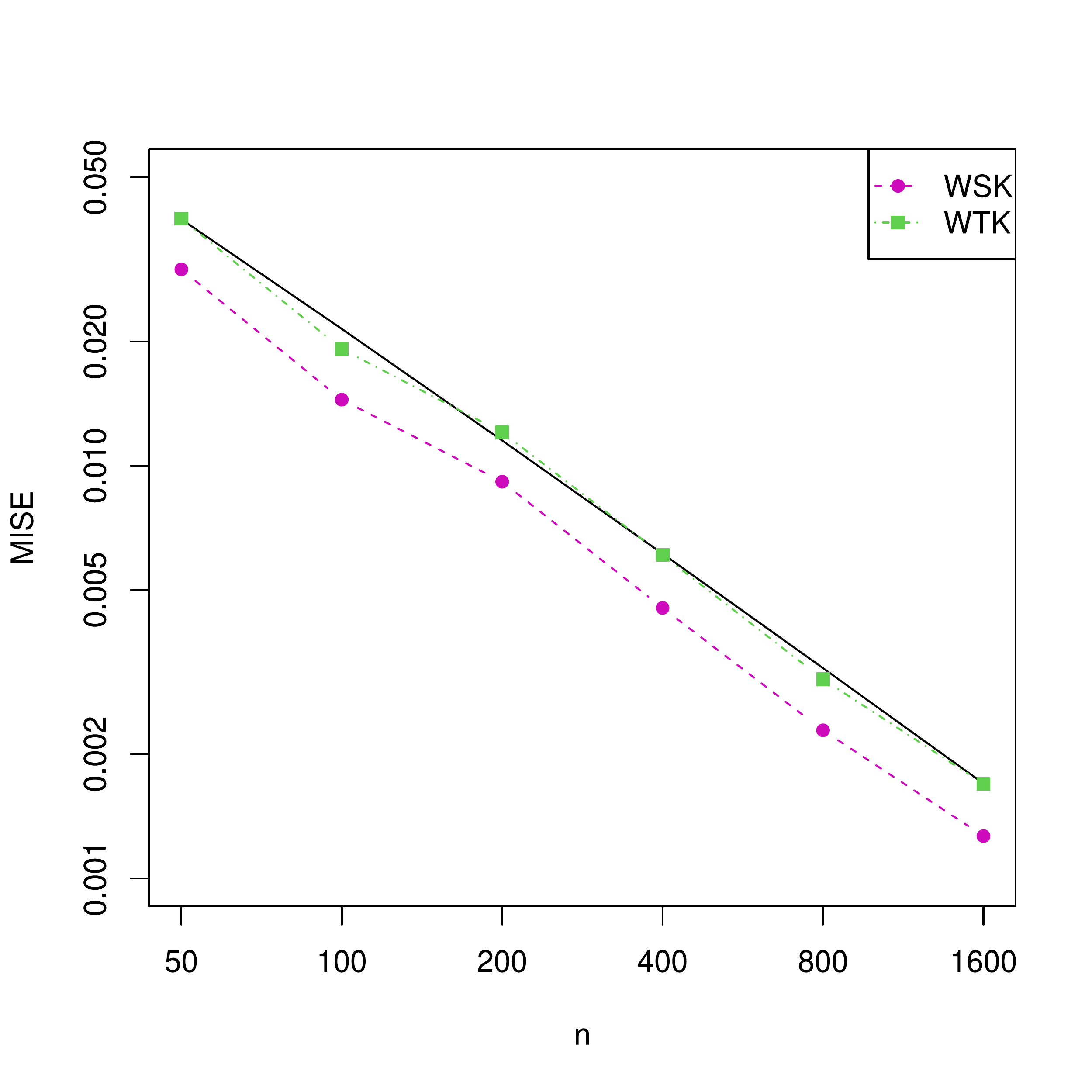}
        \subcaption{Wrapped normal distribution (M3) The solid line represents $\log{n}^{1/2}/n$, which is the order of the MISE, provided in Theorem \ref{theo: expMISE}. }
        \label{fig: wn}
      \end{minipage} &
      \begin{minipage}[t]{0.45\hsize}
        \centering
        \includegraphics[keepaspectratio, scale=0.15]{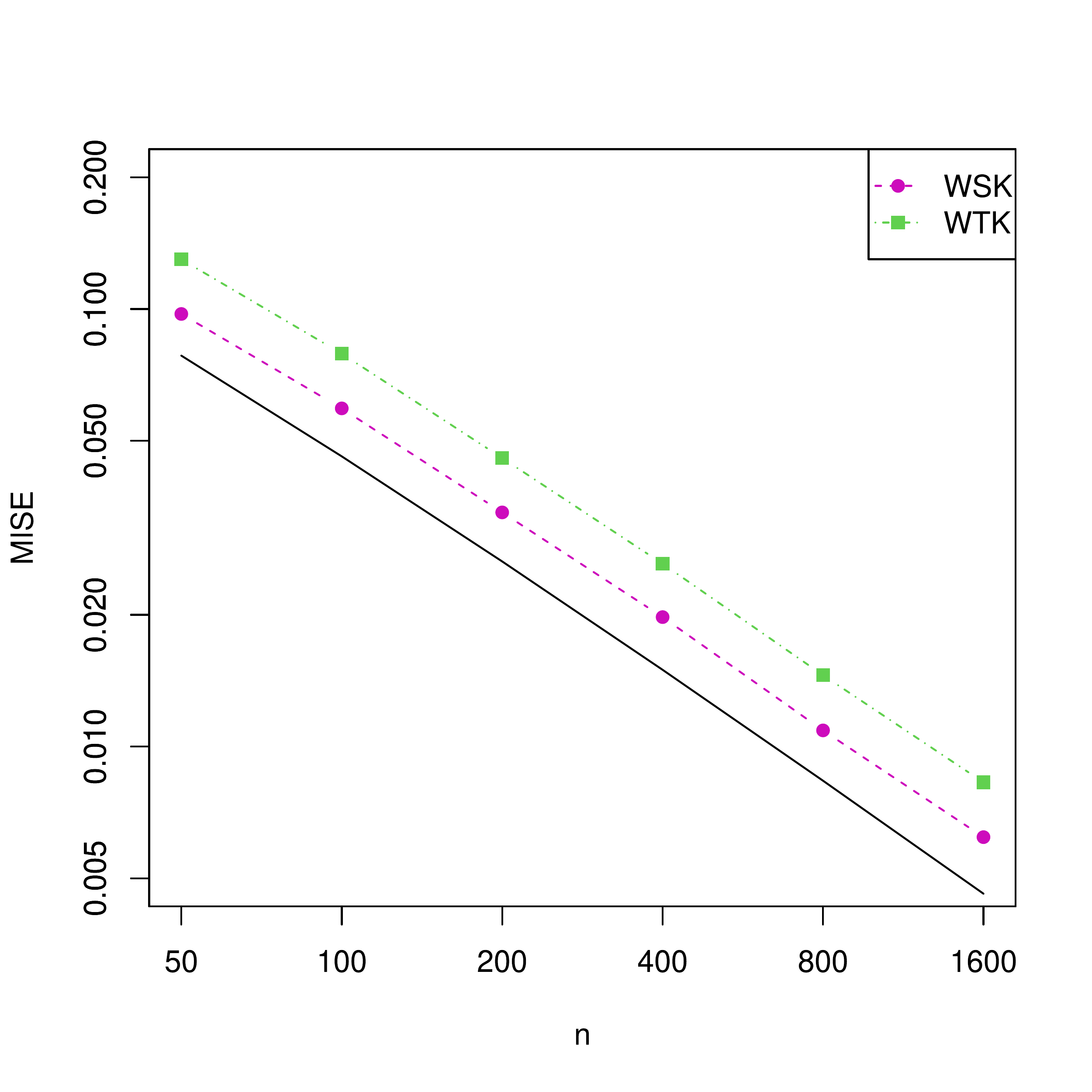}
        \subcaption{Wrapped Cauchy distribution (M5).The solid line represents $\log{n}/n$, which is the order of the MISE, provided in Theorem \ref{theo: expMISE}. }
        \label{fig: wc}
      \end{minipage} 
    \end{tabular}
     \caption{Log-log plots of $n$ (sample size) and the MISE of the wrapped sinc or wrapped trapezoid kernel density estimators based on the number of repetitions $N=1000$ with sample sizes $n=50, 100, 200, 400, 800$, and $1600$. The magenta dashed line and solid circle represent the log-log plots and MISE of the wrapped sinc kernel density estimator. The green dot-dashed line and solid box represent the log-log plots and MISE of the wrapped trapezoid kernel density estimator. The black solid line represents the log-log plots and the theoretical convergence rate of the MISE.}
\label{fig: log-log}
  \end{figure}
\subsection{Comparison with previously introduced kernel density estimators}\label{subsec: compare}
We analyze the finite sample behavior of the proposed wrapped sinc kernel, wrapped trapezoid kernel ($c=2$), and two previously introduced kernel density estimators: the von Mises kernel density estimator and fourth-order Jones and Foster kernel density estimator generated from the von Mises kernel function using the additive method proposed by \cite{tsuruta2017higher, tsuruta2024bias}. The performance measure is the MISE based on 1000 repetitions, with the sample sizes $n=1000$, $500$, $200$, and $100$. 
The scenarios are the 20 circular distributions presented by \cite{oliveira2012plug}, as shown in Table \ref{tab; scenario}.

As the smoothing parameter selectors for the wrapped flat-top kernel density estimators, we use the LSCV and ER selectors. For the LSCV selector, we select parameter $L=30$. For the von Mises kernel density estimator, we employ the LSCV, Fourier series-based plug-in (FPI), and Direct plug-in (DPI) selectors. The FPI selector is the nonparametric approach based on the empirical trigonometric moments, which achieves $\sqrt{n}$-consistency \citep{tenreiro2022kernel}. The DPI selector for circular data was introduced by \cite{DiMarzio2011kernel}. We employ the two-step DPI selector proposed by \cite{ameijeiras2024reliable}, which is available from \texttt{bw.AA.} in \texttt{NPCirc} \citep{oliveira2014npcirc}, although the-one step DPI selector has a convergence rate of $O(n^{-5/14})$ \citep{Tsuruta2020}. 

Table \ref{tab: result1} shows that when $n=1000$, the wrapped trapezoid kernel density estimator with the ER selector performs the best in six scenarios: M1, M3, M5, M8, M9, and M18. However, the performance difference between this and the wrapped sinc kernel density estimator with the ER selector in M1 and M4 is slight. The wrapped sinc kernel density estimator with the ER selector outperforms the others in multimodal scenarios: M7, M11, M12, and M16. Moreover, the wrapped sinc kernel density estimator with the ER or LSCV selectors performs worse than the wrapped trapezoid kernel density estimator with the ER estimator in M3 and M5. This result, which differs from that shown in Section \ref{subsec: optimal}, may be because the wrapped sinc kernel density estimator is susceptible to variations in smoothing parameters. Figure \ref{fig:wfk_nu} shows that as $\nu$ increases, the wrapped sinc kernel tends to
exhibit more frequent local modes away from the main mode. This behavior reflects the
fact that its characteristic function corresponds to a sharp truncation at
$|t|=\lfloor \nu \rfloor$. By contrast, the wrapped trapezoid kernel displays more stable
behavior with fewer such local features, as it retains a transition region in which the characteristic function smoothly decays to zero.

Table \ref{tab: result2} shows that when $n=500$,  the wrapped sinc and trapezoid kernel density estimators with the ER selector both outperform the previously introduced estimators in M1, M3, M4, M11, and M16. In addition, the wrapped trapezoid kernel density estimator with its selector outperforms the other estimators in M3, M5, and M8. The wrapped sinc kernel density estimator with the ER selector outperforms those in M11, M14, and M16.

Tables \ref{tab: result3} and \ref{tab: result4} show that for $n=100$ and $n=200$, the wrapped sinc and wrapped trapezoid kernel density estimators with the ER selector outperform the previously introduced estimators in M1, M4, M11, M14, and M16. In  M3, the wrapped trapezoid kernel density estimator with its selector performs the best.

\begin{figure}[htbp]
	\centering
	\begin{minipage}{0.45\linewidth}
		\centering
		\includegraphics[keepaspectratio, scale=0.38]{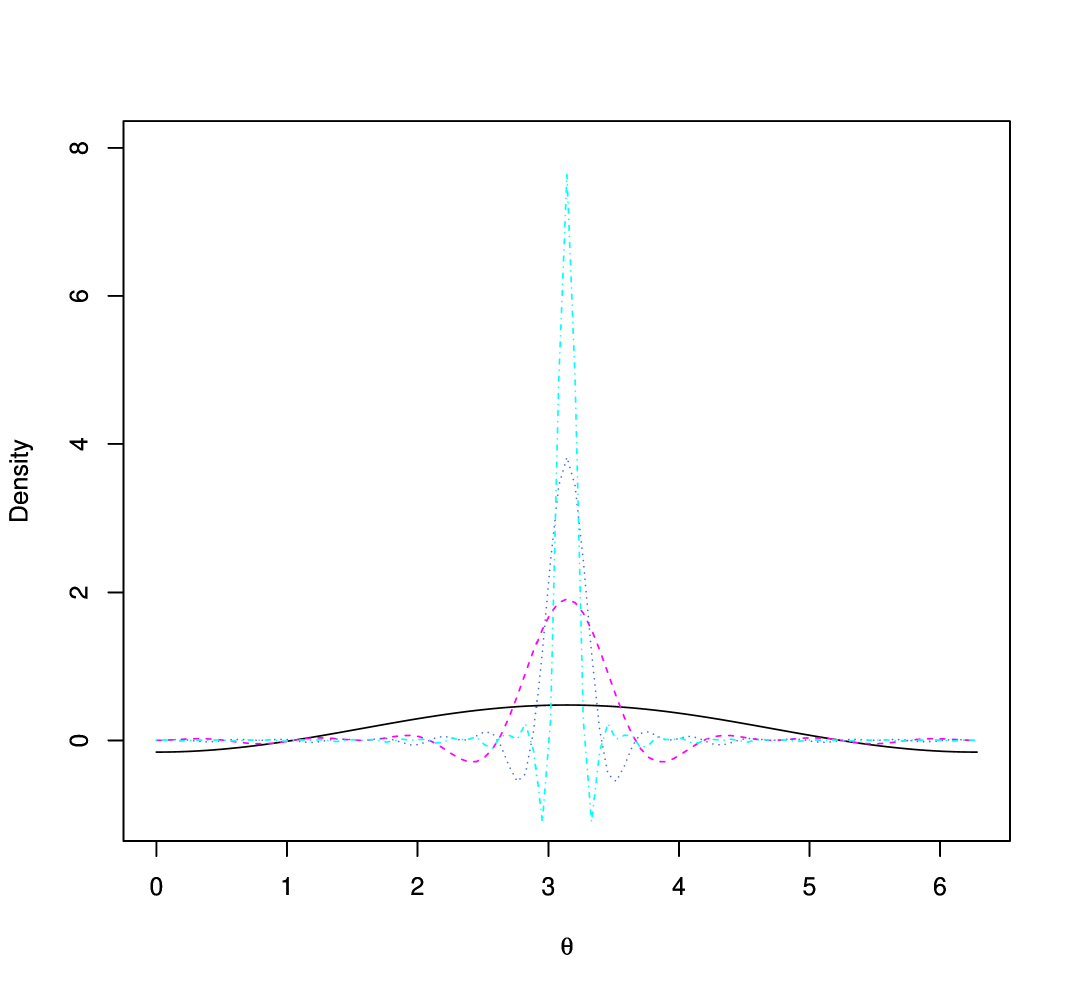}
		\subcaption{Wrapped trapezoid kernel function with mean direction $\mu=\pi$ and $c=2$.}
		\label{fig:wtk_nu}
	\end{minipage}\hfill
	\begin{minipage}{0.45\linewidth}
		\centering
		\includegraphics[keepaspectratio, scale=0.38]{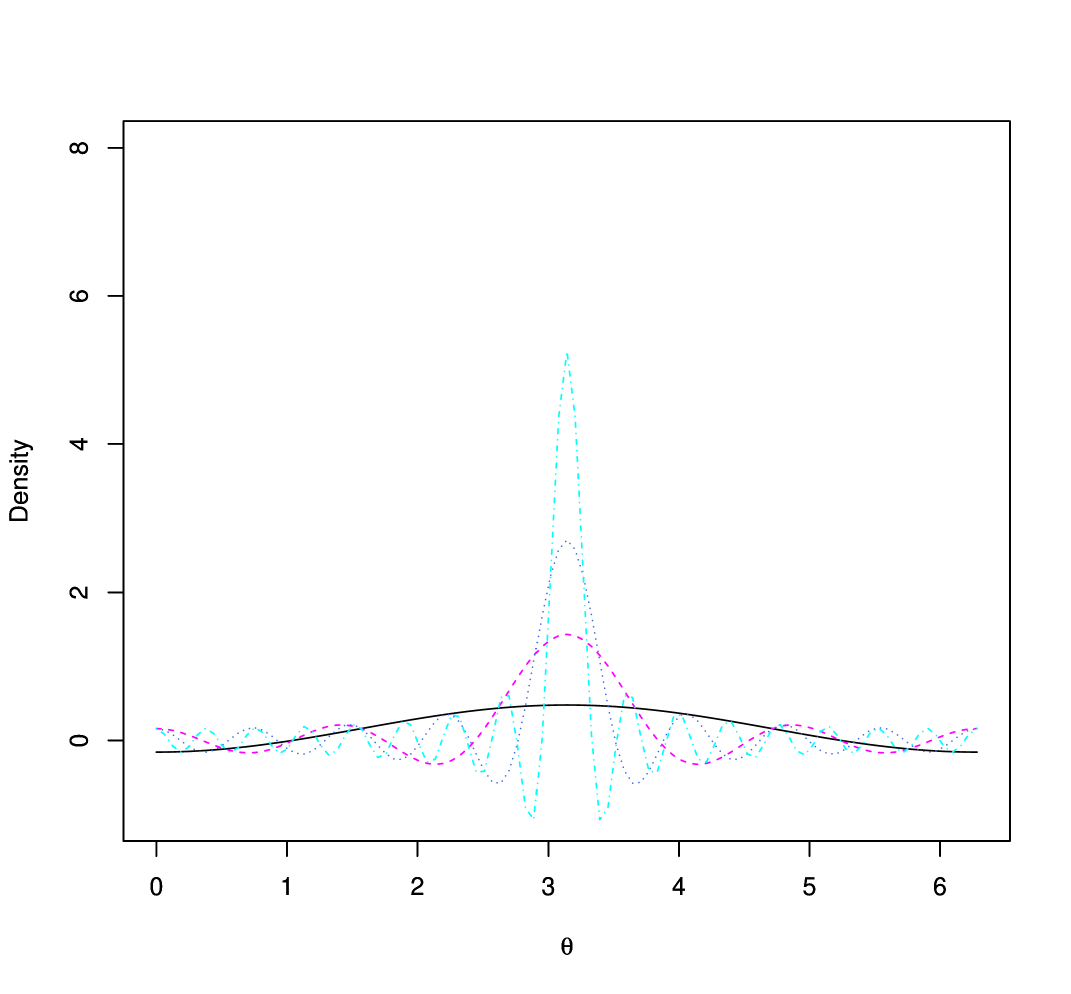}
		\subcaption{Wrapped sinc kernel function with mean direction $\mu=\pi$.}
		\label{fig:wsk_nu}
	\end{minipage}
	\caption{Effect of varying the smoothing parameter $\nu$, with $\nu = 2$ (black solid line),
		4 (magenta dashed line), 8 (blue dotted line), and 16 (cyan dot-dash line), for the wrapped
		trapezoid (left) and wrapped sinc (right) kernel functions.}
	\label{fig:wfk_nu}
\end{figure}
\begin{sidewaystable}[htbp]
  \centering
  \caption{The mean or standard error of the integrated squared error $\times 10^{4}$ of the kernel density estimators based on the number of repetitions $N=1000$ in Scenarios M1--M20. The sample size is $n=1000$. VM represents the von Mises kernel density estimator (second-order). JF4 represents the fourth-order Jones and Foster kernel density estimators based on the von Mises kernel, respectively (previously introduced estimators).  WS and WT represent the wrapped sinc and wrapped trapezoid kernel density estimators, respectively (wrapped flat-top kernel density estimators). LSCV, FPI, DPI, and ER denote the smoothing parameter selector. 
  %discussed in Section \ref{subsec: compare}. 
  }
    \begin{tabular*}{\textheight}{@{\extracolsep\fill}crrrrrrrr}
  
    \toprule
          & \multicolumn{3}{c}{VM} & \multicolumn{1}{c}{JF4} & \multicolumn{2}{c}{WS} & \multicolumn{2}{c}{WT} \\
\cmidrule{2-4}\cmidrule{6-7} \cmidrule{8-9}     No.   & \multicolumn{1}{c}{LSCV} & \multicolumn{1}{c}{FPI} & \multicolumn{1}{c}{DPI} & \multicolumn{1}{c}{LSCV} & \multicolumn{1}{c}{LSCV} & \multicolumn{1}{c}{ER} & \multicolumn{1}{c}{LSCV} & \multicolumn{1}{c}{ER} \\
    \midrule
    1     & 3.73(0.24) & 1.57(0.10) & 1.59(0.06) & 5.21(0.31) & 3.73(0.30) & 0.14(0.07) & 5.96(0.24) & 0.14(0.07) \\
    2     & 12.88(0.30) & 10.28(0.20) & 10.21(0.19) & 13.11(0.38) & 11.61(0.44) & 10.45(0.37) & 10.05(0.35) & 10.42(0.38) \\
    3     & 27.23(0.53) & 23.72(0.49) & 23.51(0.45) & 23.49(0.64) & 25.48(0.98) & 22.28(0.34) & 21.25(0.80) & 15.16(0.41) \\
    4     & 12.04(0.32) & 8.73(0.16) & 9.03(0.15) & 11.71(0.43) & 7.11(0.43) & 2.41(0.12) & 6.07(0.34) & 2.42(0.13) \\
    5     & 54.19(0.84) & 55.03(0.82) & 55.91(0.92) & 51.08(0.78) & 64.70(1.06) & 78.85(0.84) & 56.17(0.97) & 49.01(0.75) \\
    6     & 56.68(0.61) & 56.50(0.61) & 62.52(0.57) & 57.92(0.64) & 78.34(0.89) & 107.72(0.95) & 65.64(0.84) & 68.51(0.73) \\
    7     & 22.24(0.34) & 20.01(0.29) & 22.27(0.30) & 18.3(0.36) & 17.89(0.54) & 13.22(0.31) & 19.75(0.40) & 15.53(0.33) \\
    8     & 25.49(0.38) & 22.89(0.33) & 22.86(0.33) & 21.09(0.44) & 21.89(0.53) & 17.34(0.37) & 20.42(0.43) & 15.83(0.32) \\
    9     & 16.44(0.35) & 13.44(0.24) & 13.16(0.22) & 15.79(0.44) & 15.73(0.51) & 17.92(0.36) & 15.96(0.40) & 12.91(0.23) \\
    10    & 57.82(0.59) & 57.42(0.60) & 66.14(0.58) & 59.28(0.61) & 79.40(0.83) & 109.25(0.93) & 66.77(0.79) & 71.30(0.71) \\
    11    & 25.46(0.33) & 23.85(0.30) & 35.65(0.40) & 20.55(0.38) & 22.27(0.43) & 17.73(0.23) & 19.90(0.36) & 19.91(0.26) \\
    12    & 19.93(0.29) & 18.28(0.25) & 19.40(0.27) & 17.81(0.34) & 16.36(0.42) & 13.71(0.34) & 18.70(0.35) & 15.76(0.33) \\
    13    & 33.74(0.38) & 32.39(0.35) & 50.84(0.44) & 30.57(0.41) & 37.52(0.57) & 45.77(0.51) & 31.53(0.45) & 31.10(0.42) \\
    14    & 33.83(0.37) & 32.40(0.34) & 148.35(3.09) & 27.52(0.38) & 31.03(0.47) & 27.13(0.27) & 28.19(0.39) & 31.21(0.30) \\
    15    & 18.08(0.24) & 17.46(0.26) & 24.31(0.27) & 17.54(0.26) & 19.86(0.35) & 23.35(0.53) & 18.51(0.26) & 21.71(0.52) \\
    16    & 39.25(0.39) & 37.78(0.37) & 580.66(4.82) & 31.78(0.38) & 33.10(0.43) & 28.78(0.30) & 31.20(0.40) & 33.65(0.35) \\
    17    & 73.88(0.87) & 73.93(0.82) & 127.29(1.26) & 73.99(0.79) & 92.04(0.76) & 129.90(1.22) & 79.68(0.81) & 79.98(0.80) \\
    18    & 43.15(0.51) & 40.78(0.46) & 52.01(0.62) & 37.78(0.50) & 37.88(0.69) & 37.11(0.52) & 36.09(0.57) & 34.81(0.43) \\
    19    & 50.91(0.50) & 49.91(0.48) & 98.69(0.47) & 48.90(0.51) & 52.26(0.63) & 93.22(1.72) & 49.22(0.55) & 80.79(1.64) \\
    20    & 69.54(0.52) & 68.94(0.52) & 375.68(1.73) & 70.34(0.54) & 87.84(0.64) & 94.52(0.85) & 76.14(0.61) & 77.27(0.71) \\
       \bottomrule
    \end{tabular*}%
  \label{tab: result1}%
\end{sidewaystable}

\begin{sidewaystable}[htbp]
  \centering
  \caption{The mean or standard error of the integrated squared error $\times 10^{4}$ of the kernel density estimators based on the number of repetitions $N=1000$ in Scenarios M1--M20. The sample size is $n=500$. VM represents the von Mises kernel density estimator (second-order). JF4 represents the fourth-order Jones and Foster kernel density estimators based on the von Mises kernel, respectively (previously introduced estimators).  WS and WT represent the wrapped sinc and wrapped trapezoid kernel density estimators, respectively (wrapped flat-top kernel density estimators). LSCV, FPI, DPI, and ER denote the smoothing parameter selector.
  % discussed in Section \ref{subsec: compare}.
  }
    \begin{tabular*}{\textheight}{@{\extracolsep\fill}crrrrrrrr}
         \toprule
          & \multicolumn{3}{c}{VM} & \multicolumn{1}{c}{JF4} & \multicolumn{2}{c}{WS} & \multicolumn{2}{c}{WT} \\
\cmidrule{2-4}\cmidrule{6-7}   \cmidrule{8-9}  No.   & \multicolumn{1}{c}{LSCV} & \multicolumn{1}{c}{FPI} & \multicolumn{1}{c}{DPI} & \multicolumn{1}{c}{LSCV} & \multicolumn{1}{c}{LSCV} & \multicolumn{1}{c}{ER} & \multicolumn{1}{c}{LSCV} & \multicolumn{1}{c}{ER} \\
    \midrule
    1     & 9.07(0.62) & 3.29(0.21) & 3.30(0.14) & 13.00(0.81) & 7.43(0.60) & 0.76(0.58) & 12.45(0.23) & 0.81(0.62) \\
    2     & 23.19(0.62) & 17.95(0.39) & 17.49(0.34) & 26.08(0.78) & 24.75(0.51) & 28.54(1.04) & 22.05(0.54) & 28.92(0.76) \\
    3     & 48.69(1.07) & 41.14(0.92) & 40.53(0.82) & 45.45(1.35) & 53.03(1.10) & 35.56(2.04) & 47.51(0.64) & 27.15(1.67) \\
    4     & 20.16(0.53) & 14.81(0.29) & 15.64(0.29) & 18.69(0.75) & 13.7(0.48) & 4.97(0.84) & 11.72(0.19) & 5.00(0.63) \\
    5     & 97.23(1.57) & 96.37(1.57) & 99.11(1.69) & 94.15(1.59) & 120.9(2.05) & 141.46(2.18) & 104.78(1.62) & 91.28(2.06) \\
    6     & 91.78(1.05) & 92.86(1.10) & 98.19(0.97) & 98.20(1.20) & 136.05(1.17) & 171.44(1.70) & 113.04(1.56) & 119.02(1.64) \\
    7     & 39.88(0.67) & 35.53(0.56) & 40.03(0.56) & 36.16(0.81) & 36.02(0.83) & 34.79(1.03) & 35.73(0.71) & 39.02(0.78) \\
    8     & 43.82(0.74) & 39.00(0.62) & 38.95(0.60) & 38.21(0.87) & 38.98(0.83) & 40.29(1.15) & 36.67(0.80) & 31.57(0.93) \\
    9     & 28.64(0.64) & 23.12(0.44) & 22.27(0.40) & 29.37(0.81) & 30.8(0.47) & 34.98(0.91) & 27.07(0.61) & 25.11(0.65) \\
    10    & 94.26(1.05) & 95.14(1.12) & 103.35(1.01) & 101.84(1.20) & 140.71(1.21) & 175.46(1.65) & 116.42(1.64) & 125.23(1.6) \\
    11    & 43.48(0.58) & 39.85(0.50) & 70.38(0.98) & 36.61(0.64) & 35.12(0.79) & 26.62(0.80) & 33.20(0.41) & 30.30(0.61) \\
    12    & 36.11(0.62) & 32.26(0.48) & 34.37(0.48) & 35.33(0.78) & 35.38(0.56) & 38.53(0.99) & 34.49(0.83) & 40.36(0.71) \\
    13    & 58.03(0.70) & 55.39(0.65) & 95.37(0.90) & 54.99(0.77) & 68.86(0.74) & 84.44(1.02) & 58.61(0.77) & 62.54(0.84) \\
    14    & 57.69(0.68) & 54.29(0.61) & 279.86(4.40) & 49.23(0.73) & 51.64(0.92) & 41.10(1.04) & 48.61(0.51) & 47.63(0.81) \\
    15    & 32.09(0.50) & 37.96(0.55) & 38.43(0.37) & 33.95(0.63) & 40.21(0.44) & 55.71(0.77) & 35.16(0.85) & 54.90(0.61) \\
    16    & 67.43(0.70) & 64.37(0.66) & 692.93(3.02) & 57.51(0.74) & 54.39(0.92) & 45.44(0.99) & 54.39(0.63) & 53.81(0.77) \\
    17    & 124.32(1.45) & 128.43(1.41) & 209.94(2.07) & 129.25(1.32) & 168.54(1.53) & 226.30(1.51) & 143.50(2.17) & 149.00(1.50) \\
    18    & 74.45(0.92) & 70.75(0.85) & 100.78(1.13) & 70.13(1.00) & 73.72(1.04) & 75.85(1.42) & 69.37(1.02) & 67.91(1.09) \\
    19    & 85.25(0.84) & 97.20(1.09) & 138.09(0.64) & 87.85(0.97) & 107.32(0.94) & 171.26(1.41) & 93.71(1.20) & 155.25(1.15) \\
    20    & 115.85(0.90) & 119.72(0.96) & 513.95(2.21) & 120.91(0.95) & 149.85(0.91) & 167.91(1.27) & 132.03(1.23) & 143.06(1.10) \\
    \bottomrule
    \end{tabular*}%
  \label{tab: result2}%
\end{sidewaystable}

\begin{sidewaystable}[htbp]
	\centering
	\caption{The mean or standard error of the integrated squared error $\times 10^{4}$ of the kernel density estimators based on the number of repetitions $N=1000$ in Scenarios M1--M20. The sample size is $n=200$. VM represents the von Mises kernel density estimator (second-order). JF4 represents the fourth-order Jones and Foster kernel density estimators based on the von Mises kernel, respectively (previously introduced estimators).  WS and WT represent the wrapped sinc and wrapped trapezoid kernel density estimators, respectively (wrapped flat-top kernel density estimators). LSCV, FPI, DPI, and ER denote the smoothing parameter selector.
	%discussed in Section \ref{subsec: compare}. 
	}
	\begin{tabular*}{\textheight}{@{\extracolsep\fill}crrrrrrrr}
		
		\toprule
		& \multicolumn{3}{c}{VM} & \multicolumn{1}{c}{JF4} & \multicolumn{2}{c}{WS} & \multicolumn{2}{c}{WT} \\
		\cmidrule{2-4}\cmidrule{6-7} \cmidrule{8-9}     No.   & \multicolumn{1}{c}{LSCV} & \multicolumn{1}{c}{FPI} & \multicolumn{1}{c}{DPI} & \multicolumn{1}{c}{LSCV} & \multicolumn{1}{c}{LSCV} & \multicolumn{1}{c}{ER} & \multicolumn{1}{c}{LSCV} & \multicolumn{1}{c}{ER} \\
		\midrule
	1     & 19.47(1.28) & 8.47(0.54) & 8.57(0.36) & 27.05(1.68) & 20.61(1.65) & 3.92(0.73) & 30.55(1.29) & 4.13(0.77) \\
	2     & 52.23(1.63) & 37.95(0.96) & 35.26(0.73) & 58.83(2.09) & 58.81(2.17) & 53.04(0.91) & 55.01(1.71) & 54.2(1.02) \\
	3     & 107.74(2.64) & 84.25(2.21) & 82.06(1.77) & 111.57(3.54) & 121.65(5.03) & 88.12(1.94) & 105.47(4.04) & 66.64(1.90) \\
	4     & 46.17(1.48) & 30.92(0.74) & 32.63(0.65) & 47.92(2.08) & 37.90(2.20) & 15.70(0.88) & 32.19(1.68) & 16.30(1.00) \\
	5     & 200.87(3.34) & 195.00(3.33) & 201.26(3.51) & 202.63(3.58) & 264.52(5.15) & 291.88(3.46) & 226.13(4.82) & 197.34(3.37) \\
	6     & 181.27(2.31) & 180.71(2.31) & 174.62(1.94) & 199.92(2.76) & 276.23(4.03) & 304.22(2.78) & 226.36(3.46) & 233.21(2.83) \\
	7     & 83.94(1.64) & 72.18(1.19) & 83.89(1.18) & 81.48(2.07) & 93.05(2.45) & 77.43(1.08) & 80.71(1.96) & 84.40(1.16) \\
	8     & 96.55(1.85) & 80.30(1.36) & 80.00(1.29) & 92.94(2.33) & 100.28(3.13) & 90.08(1.22) & 82.82(2.47) & 69.78(1.51) \\
	9     & 64.06(1.84) & 48.99(1.06) & 44.78(0.86) & 74.03(2.23) & 81.06(2.81) & 87.23(1.44) & 65.46(2.16) & 81.35(1.70) \\
	10    & 182.90(2.48) & 188.44(2.69) & 177.21(1.91) & 202.94(2.96) & 279.13(4.30) & 321.76(2.96) & 232.82(3.76) & 264.62(3.23) \\
	11    & 93.61(1.50) & 81.16(1.09) & 159.24(2.15) & 84.98(1.84) & 76.85(2.20) & 55.91(1.14) & 74.40(1.84) & 64.95(1.24) \\
	12    & 75.32(1.42) & 67.33(0.94) & 69.97(0.94) & 78.67(1.81) & 93.32(2.10) & 89.06(1.01) & 78.42(1.50) & 90.05(1.13) \\
	13    & 121.52(1.75) & 112.62(1.45) & 209.07(2.13) & 120.59(1.96) & 159.18(2.71) & 161.19(1.82) & 130.36(2.16) & 130.80(1.85) \\
	14    & 122.00(1.58) & 111.09(1.33) & 497.16(4.84) & 109.89(1.77) & 102.94(2.39) & 82.13(1.49) & 101.21(1.90) & 97.21(1.65) \\
	15    & 65.00(1.07) & 67.55(0.60) & 59.23(0.52) & 72.22(1.32) & 93.62(1.42) & 99.98(0.78) & 76.70(1.13) & 100.60(0.80) \\
	16    & 140.54(1.65) & 130.50(1.43) & 760.83(1.26) & 127.34(1.84) & 113.05(2.50) & 93.51(1.62) & 118.35(2.01) & 113.31(1.80) \\
	17    & 244.79(2.92) & 254.71(3.22) & 364.24(3.46) & 259.71(2.91) & 349.45(3.57) & 441.16(4.56) & 291.47(3.55) & 328.55(4.25) \\
	18    & 154.70(2.06) & 170.54(2.68) & 205.44(1.71) & 158.43(2.45) & 185.99(3.75) & 227.98(3.51) & 167.77(3.04) & 212.27(3.73) \\
	19    & 170.73(1.86) & 189.07(1.44) & 198.23(1.15) & 187.16(2.17) & 251.19(2.76) & 232.46(1.69) & 211.14(2.44) & 224.92(1.73) \\
	20    & 226.90(1.85) & 226.62(1.78) & 694.93(2.49) & 242.47(2.05) & 309.03(2.69) & 296.56(1.97) & 273.22(2.46) & 267.23(2.11) \\
		\bottomrule
	\end{tabular*}%
	\label{tab: result3}%
\end{sidewaystable}
\begin{sidewaystable}[htbp]
	\centering
	\caption{The mean or standard error of the integrated squared error $\times 10^{4}$ of the kernel density estimators based on the number of repetitions $N=1000$ in Scenarios M1--M20. The sample size is $n=100$. VM represents the von Mises kernel density estimator (second-order). JF4 represents the fourth-order Jones and Foster kernel density estimators based on the von Mises kernel, respectively (previously introduced estimators).  WS and WT represent the wrapped sinc and wrapped trapezoid kernel density estimators, respectively (wrapped flat-top kernel density estimators). LSCV, FPI, DPI, and ER denote the smoothing parameter selector.
	% discussed in Section \ref{subsec: compare}. 
	}
	\begin{tabular*}{\textheight}{@{\extracolsep\fill}crrrrrrrr}
		
		\toprule
		& \multicolumn{3}{c}{VM} & \multicolumn{1}{c}{JF4} & \multicolumn{2}{c}{WS} & \multicolumn{2}{c}{WT} \\
		\cmidrule{2-4}\cmidrule{6-7} \cmidrule{8-9}     No.   & \multicolumn{1}{c}{LSCV} & \multicolumn{1}{c}{FPI} & \multicolumn{1}{c}{DPI} & \multicolumn{1}{c}{LSCV} & \multicolumn{1}{c}{LSCV} & \multicolumn{1}{c}{ER} & \multicolumn{1}{c}{LSCV} & \multicolumn{1}{c}{ER} \\
		\midrule
	
	1     & 38.29(2.53) & 17.53(1.07) & 19.11(0.78) & 55.34(3.35) & 40.83(2.89) & 13.94(1.95) & 63.57(2.74) & 14.85(2.09) \\
	2     & 93.85(3.27) & 63.86(1.75) & 59.91(1.31) & 108.04(4.19) & 99.84(3.85) & 75.25(1.96) & 93.02(3.07) & 77.49(2.23) \\
	3     & 185.12(4.72) & 139.03(3.66) & 135.81(2.97) & 197.64(6.45) & 220.28(9.07) & 160.56(2.87) & 191.33(8.19) & 120.16(3.65) \\
	4     & 88.13(3.00) & 53.94(1.54) & 57.46(1.27) & 97.50(4.08) & 80.01(4.73) & 41.64(2.43) & 66.55(3.56) & 44.23(2.73) \\
	5     & 355.48(6.26) & 342.96(6.60) & 345.20(6.19) & 369.87(7.07) & 488.82(10.42) & 507.42(6.37) & 416.89(9.80) & 364.23(6.77) \\
	6     & 290.89(4.39) & 281.44(3.93) & 254.4(3.00) & 321.17(5.38) & 439.86(8.07) & 439.13(4.22) & 360.83(6.86) & 357.70(4.85) \\
	7     & 152.89(3.44) & 122.85(2.16) & 147.03(2.07) & 154.67(4.32) & 159.95(4.62) & 121.80(2.37) & 144.18(4.10) & 133.08(2.64) \\
	8     & 171.48(3.81) & 138.89(2.59) & 138.97(2.37) & 173.42(4.81) & 178.51(5.38) & 150.82(2.87) & 159.58(5.48) & 133.98(3.41) \\
	9     & 113.29(3.23) & 88.63(1.84) & 77.81(1.49) & 131.18(4.06) & 147.95(4.84) & 135.04(1.85) & 128.13(3.88) & 134.41(2.20) \\
	10    & 290.63(4.35) & 295.95(3.74) & 254.88(2.91) & 328.10(5.50) & 447.36(7.95) & 439.11(3.12) & 369.20(6.70) & 393.3(3.76) \\
	11    & 168.02(3.11) & 142.49(2.08) & 261.84(2.69) & 162.65(3.80) & 149.79(4.74) & 114.10(2.76) & 145.29(3.68) & 134.15(3.11) \\
	12    & 133.54(2.82) & 114.27(1.85) & 119.95(1.75) & 143.75(3.49) & 161.32(3.60) & 144.54(2.74) & 140.97(3.09) & 148.97(2.93) \\
	13    & 212.37(3.37) & 188.02(2.48) & 346.74(3.36) & 218.24(4.16) & 268.3(5.28) & 266.05(2.88) & 224.51(4.21) & 225.06(3.45) \\
	14    & 220.24(3.59) & 193.26(2.68) & 637.57(3.82) & 209.06(4.25) & 192.06(5.45) & 156.69(3.47) & 198.91(4.37) & 185.97(3.82) \\
	15    & 105.20(2.36) & 83.56(1.04) & 76.63(0.78) & 121.55(3.02) & 140.91(3.36) & 123.20(1.61) & 118.85(2.66) & 125.99(1.83) \\
	16    & 247.67(3.19) & 226.12(2.61) & 786.29(0.50) & 237.37(3.91) & 205.65(4.19) & 182.11(3.37) & 236.19(3.89) & 221.95(3.75) \\
	17    & 414.27(5.41) & 456.58(7.10) & 519.06(5.27) & 449.27(5.73) & 609.53(7.29) & 717.65(7.47) & 513.08(7.22) & 606.00(8.64) \\
	18    & 268.16(3.67) & 322.05(2.95) & 298.47(2.13) & 292.79(4.66) & 377.83(6.52) & 410.23(3.25) & 326.74(5.96) & 406.59(3.67) \\
	19    & 277.44(3.40) & 264.96(2.82) & 262.04(1.98) & 305.16(4.24) & 372.58(6.00) & 329.05(4.19) & 341.74(5.51) & 321.98(4.22) \\
	20    & 368.73(3.37) & 352.64(3.19) & 814.26(2.65) & 394.79(3.83) & 486.01(5.87) & 445.60(4.32) & 433.88(4.56) & 425.48(4.70) \\
		\bottomrule
	\end{tabular*}%
	\label{tab: result4}%
\end{sidewaystable}
\section{Application}\label{sec: application}
We present three empirical examples showing the utility of the wrapped sinc and trapezoid kernel density estimators with the ER selector.  Section \ref{sec: turtles} summarizes data on turtles. Section \ref{sec: wind} shows wind data.  Section \ref{sec: zebrafish} illustrates zebrafish data. These data are available from \texttt{R}.
\subsection{Turtles data}\label{sec: turtles}
The data on turtles are available as \texttt{fisherB3} in the \texttt{circular} library \citep{circular2022}. These data are the directions taken by 76 turtles after treatment \citetext{\citealp[]{stephens1969techniques}; \citealp[p. 241]{fisher1995statistical}}.  Figure \ref{fig: turtles} shows that the wrapped flat-top kernel density estimators represent the two-modal density for its data. \cite{mardia2000directional} note that the data indicate that the turtles have a preferred direction (toward the sea); however, a substantial minority seem to prefer the opposite direction. Hence, the  wrapped flat-top kernel density estimators effectively capture the data characteristics.

The wrapped flat-top kernel density estimators take negligibly low negative values at points which the data are coarse. For these negative values, the wrapped trapezoid kernel density estimator takes values closer to 0 than the wrapped sinc kernel density estimator. These negative values can be easily removed by $\max(0, \hat{f}_{\nu})$.  Therefore, they do not affect the interpretation of the data. To ensure that the estimator is a proper density, we can use the method of \cite{glad2003correction}, which adds an adjustment term to $\max(0, \hat{f}_{\nu})$ so that the integral value is 1.
 \begin{figure}[htbp]
    \begin{tabular}{cc}
        \begin{minipage}[t]{0.45\hsize}
        \centering
        \includegraphics[keepaspectratio, scale=0.45]{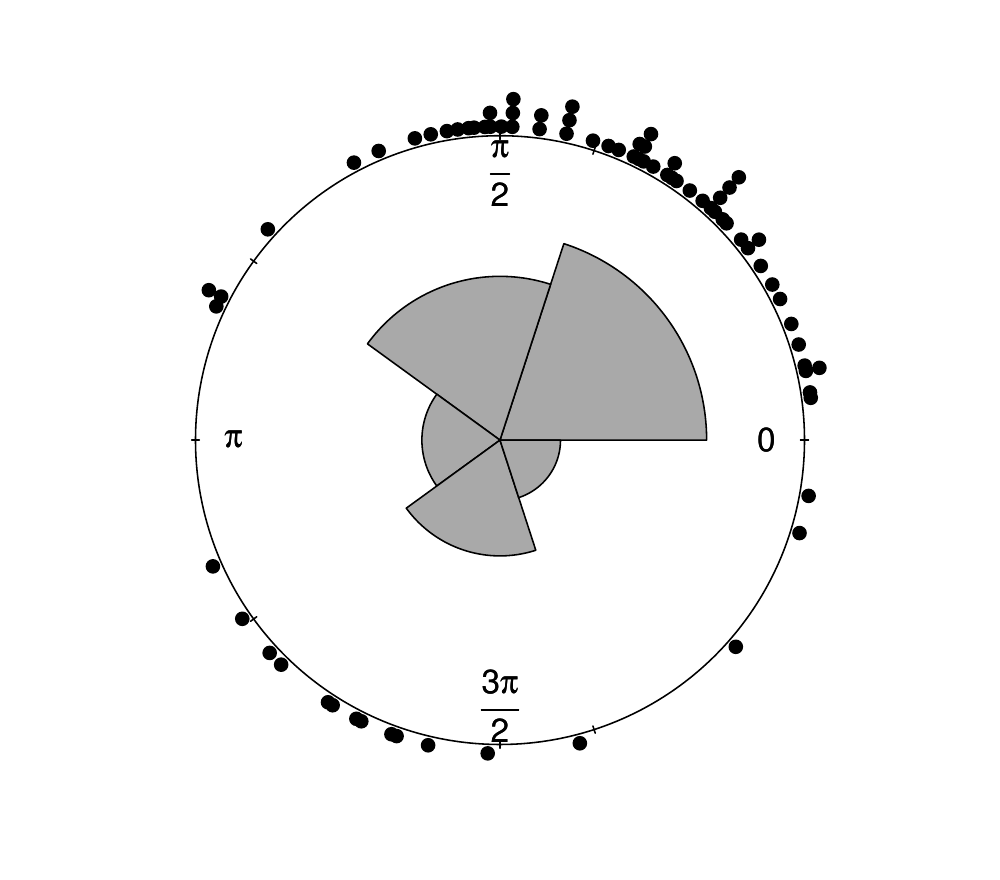}
        \subcaption{Circular plot and Rose diagram with bin width $\delta = 1.26$ estimated by the equal-split biased cross-validation proposed by \cite{tsuruta2023automatic}. }
        \label{fig: turtles_rose}
      \end{minipage}&
      \begin{minipage}[t]{0.45\hsize}
        \centering
        \includegraphics[keepaspectratio, scale=0.45]{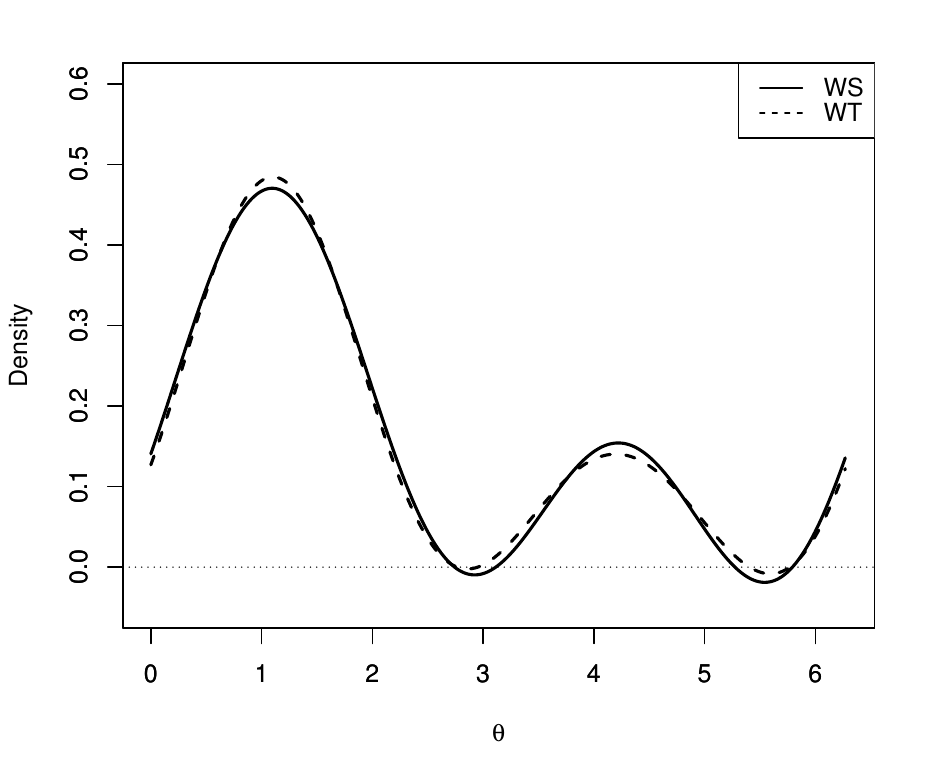}
        \subcaption{The solid line and dashed line represent the wrapped sinc and wrapped trapezoid kernel density estimators, respectively. Both smoothing parameters are $\nu = 2$, estimated by the ER selector. }
        \label{fig: turtles_density}
      \end{minipage} \\
    \end{tabular}
     \caption{Turtles data.}
\label{fig: turtles}
  \end{figure}
\subsection{Wind data}\label{sec: wind}
 We use wind data from \texttt{wind} in the \texttt{circular} library. The data consist of 310 observations recorded from January 29, 2001 to March 31, 2001, at the Col de la Roa weather station in the Italian Alps, with the values measured in radians, clockwise from the north \citep{circular2022}.
 
As shown in Figure \ref{fig: wind}, the wrapped flat-top kernel density estimators represent the dominant mode and many small modes. 
 The wrapped trapezoid kernel density estimator provides a density that has a larger peak and more moderate zigzags than the wrapped sinc kernel density estimator. Furthermore, it rarely takes negative values, although the wrapped sinc kernel density estimator sometimes takes negligibly low negative values.
 
 \begin{figure}[htbp]
    \begin{tabular}{cc}
        \begin{minipage}[t]{0.45\hsize}
        \centering
        \includegraphics[keepaspectratio, scale=0.40]{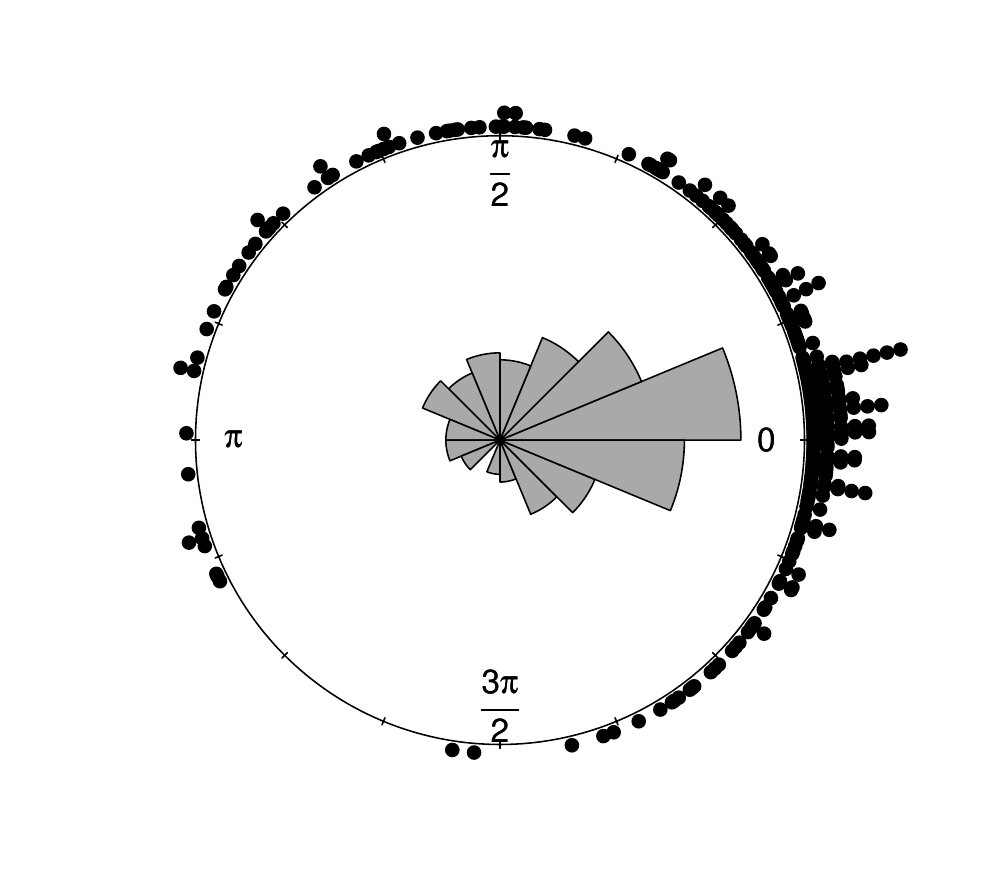}
        \subcaption{Circular plot and Rose diagram with bin width $\delta = 0.39$ estimated by the equal-split biased cross-validation. }
        \label{fig: wind_rose}
      \end{minipage}&
      \begin{minipage}[t]{0.40\hsize}
        \centering
        \includegraphics[keepaspectratio, scale=0.40]{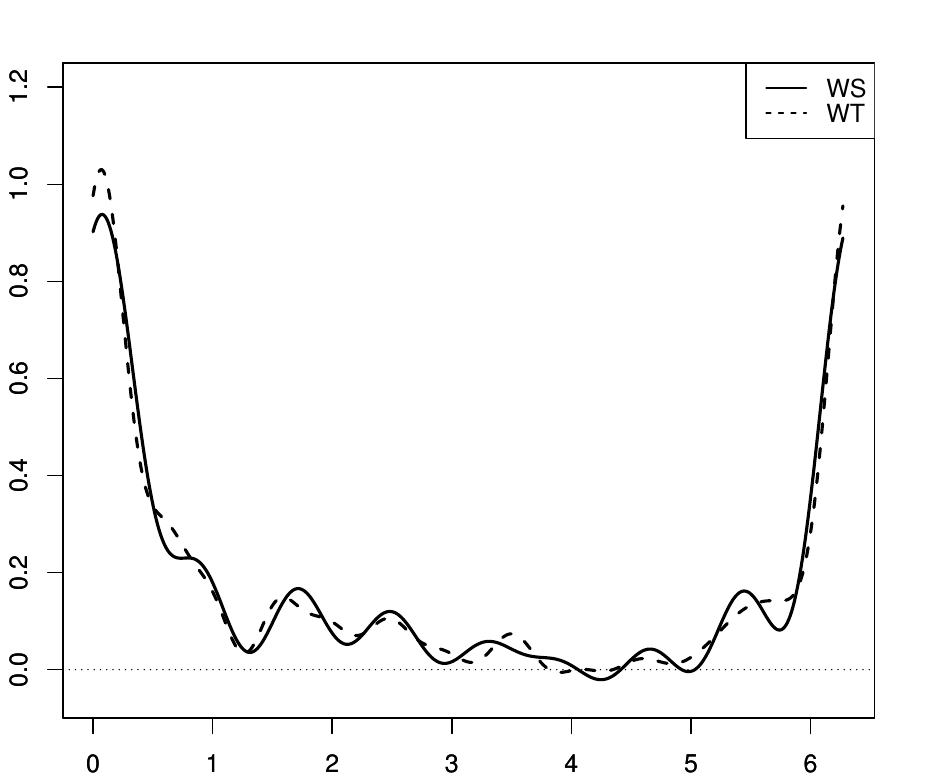}
        \subcaption{The solid line and dashed line represent the wrapped sinc and wrapped trapezoid kernel density estimator, respectively. Both smoothing parameters are $\nu = 8$,  estimated by the ER selector. }
        \label{fig: wind_density}
      \end{minipage} \\
    \end{tabular}
     \caption{Wind data.}
\label{fig: wind}
  \end{figure}
 \subsection{Zebrafish data} \label{sec: zebrafish}
 The zebrafish data are available as the \texttt{res\_angle} variable in \texttt{zebrafish} in the \texttt{NPCirc} library.  The dataset includes 502 observations corresponding to the directions of escape of each fish, obtained from an experimental study on larval zebrafish, which were startled by a predator imitated by a robot disguised as an adult zebrafish \citep{alonso2023analyzing}.
 
As shown in Figure \ref{fig: zebrafish}, the wrapped flat-top kernel density estimators represent a density that has a dominant mode and one small mode. They do not take negative values.  \cite{alonso2023analyzing} state that ``the larval zebrafish have one preferred direction of escape when being chased from a direction lateral to their bodies. However, when the zebrafish are pursued from the caudal or rostral sides of their bodies, there are two preferred directions of escape.'' The wrapped flat-top kernel density estimators effectively illustrate these characteristics.
 \begin{figure}[htbp]
    \begin{tabular}{cc}
        \begin{minipage}[t]{0.45\hsize}
        \centering
        \includegraphics[keepaspectratio, scale=0.40]{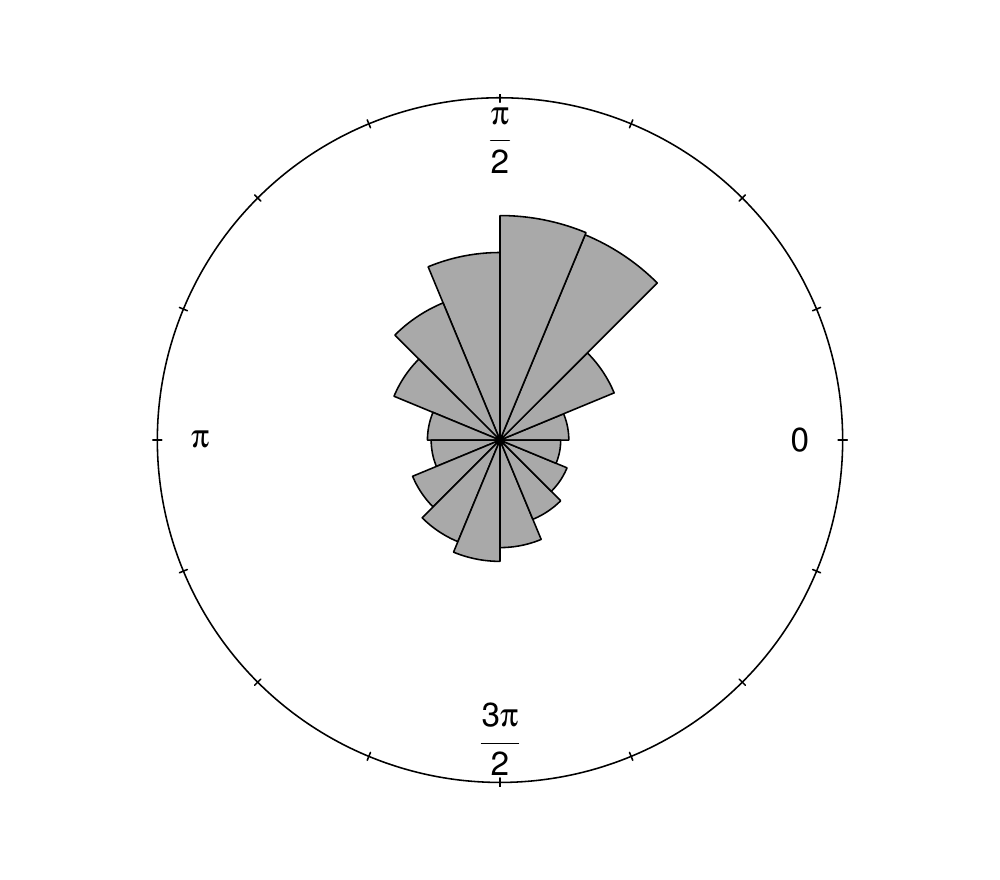}
        \subcaption{Circular plot and rose diagram with bin width $\delta =  0.39$ estimated by the equal-split biased cross-validation. }
        \label{fig: zebrafish_rose}
      \end{minipage}&
      \begin{minipage}[t]{0.40\hsize}
        \centering
        \includegraphics[keepaspectratio, scale=0.40]{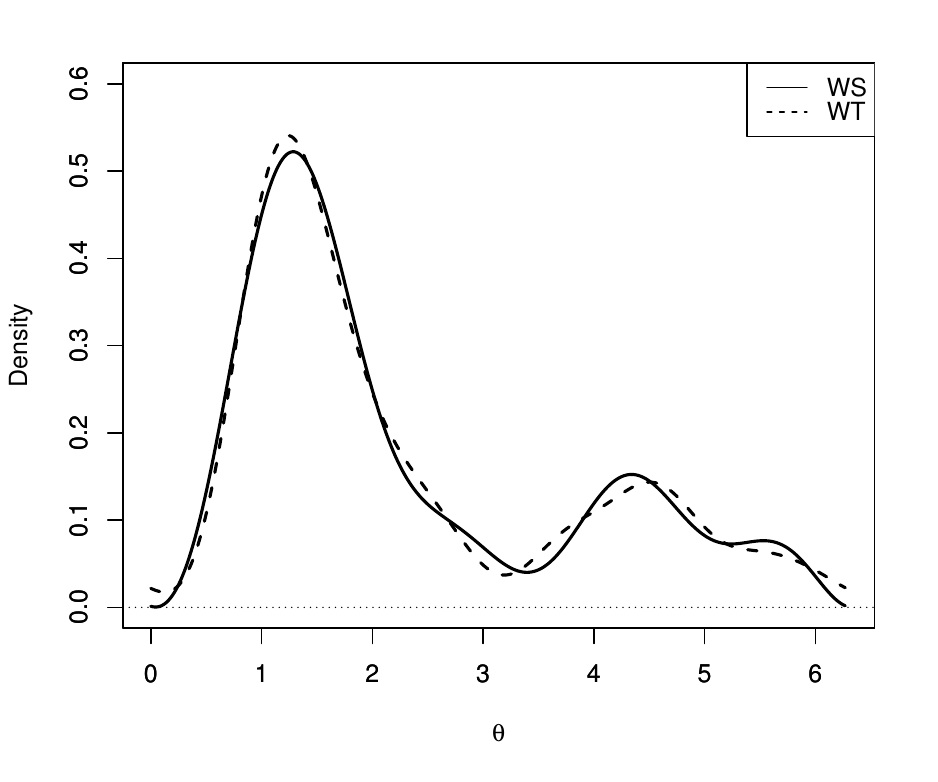}
        \subcaption{The solid line and dashed line represent the wrapped sinc and wrapped trapezoid kernel density estimator, respectively.  the both smoothing parameters are $\nu = 4$, estimated by the ER selector. }
        \label{fig: zebrafish_density}
      \end{minipage} \\
    \end{tabular}
     \caption{zebrafish data.}
\label{fig: zebrafish}
  \end{figure}
\section{Conclusion}\label{sec: conclusion}
This study proposed a wrapped flat-top kernel density estimator generated from a flat-top kernel function on the real line, namely, the wrapped sinc or trapezoid kernel density estimators. The wrapped flat-top kernel density estimators achieve $\sqrt{n}$-consistency when the characteristic function of the underlying circular density $f$ has compact support, such as for circular uniform and cardioid distributions. The wrapped flat-top kernel density estimators have a convergence rate of the MISE, which is $O((\log{n})^{1/\alpha} n^{-1})$ under a squared value of a characteristic function of $f$ exponential decay with parameter $\alpha$, such as for the wrapped normal ($\alpha =2$) and wrapped Cauchy ($\alpha =1$) distributions. These estimators also have a convergence rate of the MISE, which is $O(n^{-2r/(2r+1)})$ under a squared value of a characteristic function of $f$ decay at the speed of $t^{2r}$, such as $r$ times continuously differentiable and circular triangular density.  

In summary, wrapped flat-top kernel density estimators have two advantages over the previously introduced estimator. They have a faster convergence rate of the MISE and are more flexible because they can accurately estimate densities that are not smooth, such as non-differentiable and uniform densities. 

We confirm the asymptotic properties derived in the theoretical section through numerical experiments. We find that the wrapped flat-top kernel density estimators outperform the previously introduced estimators under the well-used simulation setting of circular non-parametric estimators. Specifically, the wrapped trapezoid kernel density estimator with the ER selectors is the most effective. However, the results depend on the selection of the smoothing parameter. In empirical analyses, we find that the wrapped flat-top kernel density estimators with the ER selectors illustrate the characteristics of the data, although the wrapped trapezoid kernel density provides more appropriate results than the wrapped sinc kernel density estimator.

These results show that the estimation accuracy of wrapped flat-top kernel density estimators is higher than that of the previously introduced estimators. Therefore,  wrapped flat-top kernel density estimators are expected to allow flexible and accurate estimation in circular data analysis. However, they rarely take on negative values. The empirical analysis indicates that no undesirable properties are observed when the sample size is larger than 500. Furthermore, we correct the negativity of these using $\max(0, \hat{f}_{\nu}(\theta))$.

Future studies could extend the use of these wrapped flat-top kernel density estimators to multivariate spaces, such as a torus and cylinder. Furthermore, they could investigate the properties of wrapped flat-top kernel density estimators for skewed circular and related directional distributions
 \citep{abe2011sine, ley2017skew}.  

%Furthermore, asymptotic properties of the smoothing parameter selectors, such as the ER selector.

%\bmhead{Supplementary information}
%This paper is accompanied by an \texttt{R} code file, entitled “code\_WFTKDE.R”, which implements
%the wrapped sinc and wrapped trapezoid kernel density estimators, together with the ER.
%selector.
%If your article has accompanying supplementary file/s please state so here. 

%Authors reporting data from electrophoretic gels and blots should supply the full unprocessed scans for key as part of their Supplementary information. This may be requested by the editorial team/s if it is missing.

%Please refer to Journal-level guidance for any specific requirements.

\section*{Acknowledgements}
This work was supported by JSPS KAKENHI Grant Numbers JP20K19760 and JP24K20746.
%Acknowledgements are not compulsory. Where included they should be brief. Grant or contribution numbers may be acknowledged.

%Please refer to Journal-level guidance for any specific requirements.

\begin{appendices}
\section{Proofs}\label{secA1}
\begin{proof}[Proof of Theorem \ref{theo: MISE}]
By Parseval's formula, the MISE is given by
\begin{align}
\MISE[\hat{f}_{h}(\theta)] 
&= \frac{1}{2\pi}\sum_{t\in \mathbb{Z}}\E[| \phi_{t}(f) - \phi_{t}(\hat{f}_{h})|^{2}]\notag\\
&=  \frac{1}{2\pi}\sum_{t \in \mathbb{Z}}\left(|\phi_{t}(f)|^{2} - \bar{\phi_{t}}(f)\E[\phi_{t}(\hat{f}_{h})] - \phi_{t}(f)\E[\bar{\phi}_{t}(\hat{f}_{h})] + \E[|\phi_{t}(\hat{f}_{h})|^{2}]\right), \label{eq; MISEpf}
\end{align}
where $\phi_{t}(\hat{f}_{h})= n^{-1}\phi_{t}(K_{h})\sum_{j}e^{it\Theta_{j}}$.
We obtain
\begin{align}
\bar{\phi_{t}}(f)\E[\phi_{t}(\hat{f}_{h})] &= \bar{\phi_{t}}(f)\frac{1}{n}\sum_{j=1}^{n}\phi_{t}(K_{h})\E[e^{it\Theta_{j}}]= \phi(K_{h}) |\phi_{t}(f)|^{2}, \label{eq: A1}\\
\phi_{t}(f)\E[\bar{\phi}_{t}(\hat{f}_{h})] & = \bar{\phi}(K_{h}) |\phi_{t}(f)|^{2} , \label{eq: A2}\\
\shortintertext{and}
\E[|\phi_{t}(\hat{f}_{h})|^{2}] &= \E\left[\frac{1}{n}|\phi_{t}(K_{h})|^{2} + \frac{1}{n^{2}}\sum_{j \neq k} |\phi_{t}(K_{h})|^{2}e^{it(\Theta_{j} - \Theta_{k})} \right]\notag\\
&= \frac{1}{n} |\phi_{t}(K_{h})|^{2} + \frac{n(n-1)}{n^{2}}|\phi_{t}(K_{h})|^{2}|\phi_{t}(f)|^{2}. \label{eq: A3}
\end{align}
Substituting \eqref{eq: A1}, \eqref{eq: A2}, and  \eqref{eq: A3} to the right-hand side of \eqref{eq; MISEpf} yields
\begin{align}
\MISE[\hat{f}_{h}(\theta)]&= \frac{1}{2\pi}\left[\sum_{t \in \mathbb{Z}}|\phi_{t}(f)|^{2}|1 - \phi_{t}(K_{h})|^{2} + \frac{1}{n}\sum_{t \in \mathbb{Z}}|\phi_{t}(K_{h})|^{2}(1 - |\phi_{t}(f)|^{2})\right]. \label{eq: MISEpf1}
\end{align}
\end{proof}
 \begin{proof}[Proof of Theorem \ref{theo: IVwftk}]
 We obtain 
 \begin{align}
 R(K_{\nu, c}^{2}) &= \frac{1}{2\pi}\sum_{|t|< c\floor{\nu}}|\phi_{t}(K_{\nu,c})|^{2}\notag\\
 &= \frac{1}{2\pi}\{1 + 2\floor{\nu} +  \sum_{\floor{\nu} < |t| < c \floor{\nu}}g(t; \nu, c)^{2}\}\notag\\
 &\le \frac{1 + 2(c\floor{\nu}-1)}{2\pi}\notag\\
 &\le \frac{c\nu}{\pi}.\label{eq: Rwftk}
 \end{align}
 Combining \eqref{eq: IV0} and \eqref{eq: Rwftk} yields
 \begin{align}
 \IV[\hat{f}_{\nu,c}(\theta)]&\le \frac{c\nu}{n\pi} - \frac{1}{2\pi n}\sum_{0 < |t| < c\nu}|\phi_{t}(f)|^{2}|\phi_{t}(K_{\nu,c})|^{2}\notag\\
 &\le \frac{c\nu}{n\pi}.
 \end{align} 
 \end{proof}
 \begin{proof}[Proof of Theorem \ref{theo: rMISE}]
 Combining \eqref{eq: ISB0} and Condition C1) yields
\begin{align}
\ISB[\hat{f}_{\nu,c}(\theta)]&\le\frac{1}{2\pi}\frac{1}{\nu^{2r}}\sum_{|t|>\nu}\nu^{2r}|\phi_{t}(f)|^{2}\notag\\
&\le \frac{1}{2\pi}\frac{1}{\nu^{2r}}\sum_{|t|>\nu}t^{2r}|\phi_{t}(f)|^{2}\notag\\
&=  \frac{1}{2\pi}\frac{1}{\nu^{2r}}\sum_{t\in \mathbb{Z}}t^{2r}|\phi_{t}(f)|^{2} \left( 1- \frac{\sum_{|s|\le\nu}s^{2r}|\phi_{s}(f)|^{2}}{\sum_{s\in \mathbb{Z}}s^{2r}|\phi_{s}(f)|^{2}} \right)\notag\\
&= \frac{1}{\nu^{2r}}C_{r}(f)\epsilon(\nu), %\label{eq: ISBp}
\end{align}
where 
\begin{align*}\epsilon(\nu):=  1- \frac{\sum_{|s|\le\nu}s^{2r}|\phi_{s}(f)|^{2}}{\sum_{s\in \mathbb{Z}}s^{2r}|\phi_{s}(f)|^{2}} .
%, \quad \epsilon(\nu) \in (0, 1).
\end{align*}
 \end{proof}
\begin{proof}[Proof of Theorem \ref{theo: expMISE}]
Combining \eqref{eq: ISB0} and Condition C2) yields
\begin{align}
\ISB[\hat{f}_{\nu,c}(\theta)]&\le\frac{1}{2\pi}e^{-\tau|\nu|^{\alpha}}\sum_{|t|>\nu}e^{\tau|\nu|^{\alpha}}|\phi_{t}(f)|^{2}\notag\\
&\le \frac{1}{2\pi}e^{-\tau|\nu|^{\alpha}}\sum_{|t|>\nu}e^{\tau|t|^{\alpha}}|\phi_{t}(f)|^{2}\notag\\
&=  \frac{1}{2\pi}e^{-\tau|\nu|^{\alpha}}\sum_{t\in \mathbb{Z}}e^{\tau|t|^{\alpha}}|\phi_{t}(f)|^{2} \left( 1- \frac{\sum_{|s|\le\nu}e^{\tau|s|^{\alpha}}|\phi_{s}(f)|^{2}}{\sum_{s\in \mathbb{Z}}e^{\tau|s|^{\alpha}}|\phi_{s}(f)|^{2}} \right)\notag\\
&=e^{-\tau|\nu|^{\alpha}}I_{\alpha,\tau}\epsilon_{2}(\nu), %\label{eq: ISBp}
\end{align}
where 
\begin{align*}\epsilon_{2}(\nu):= 1- \frac{\sum_{|s|\le\nu}e^{\tau|s|^{\alpha}}|\phi_{s}(f)|^{2}}{\sum_{s\in \mathbb{Z}}e^{\tau|s|^{\alpha}}|\phi_{s}(f)|^{2}} .
%, \quad \epsilon_{2}(\nu) \in (0, 1).
\end{align*}
\end{proof}




\end{appendices}

\bibliographystyle{abbrvnat}
\bibliography{WFTKDE_ST}% common bib file

\end{document}